\PassOptionsToPackage{unicode}{hyperref}
\PassOptionsToPackage{hyphens}{url}
\PassOptionsToPackage{dvipsnames,svgnames,x11names}{xcolor}
\documentclass[12pt]{article}
\usepackage{tikz}
\usetikzlibrary{shapes.geometric, arrows, positioning, fit, backgrounds, calc, arrows.meta}
\usepackage{amsmath,amssymb,amsfonts,amsthm,bm}

\usepackage{natbib}
\usepackage{authblk}
\usepackage{graphicx,psfrag,epsf}
\usepackage{enumerate}
\usepackage{booktabs}
\usepackage{multirow}

\newcommand{\argmin}{\arg\!\min}

\newcommand{\Var}{\mbox{Var}}

\newcommand{\E}{\mbox{E}}

\newcommand{\ve}[1]{\mbox{\boldmath ${#1}$}}
\newcommand{\vesub}[2]{\mbox{{\boldmath ${#1}$}$_{#2}$}}
\newcommand{\vesup}[2]{\mbox{{\boldmath ${#1}$}$^{#2}$}}
\newcommand{\vess}[3]{\mbox{{\boldmath ${#1}$}$_{#2}^{#3}$}}
\newcommand{\hve}[1]{\hat{\ve{#1}}}

\newcommand{\hvesub}[2]{\hat{\ve{#1}}_{#2}}

\newtheorem{prop}{Proposition}
\newtheorem{theorem}{Theorem}
\newtheorem{remark}{Remark}

\newtheorem{condition}{Condition}
\newtheorem{assumption}{Assumption}

\usepackage{iftex}
\ifPDFTeX
  \usepackage[T1]{fontenc}
  \usepackage[utf8]{inputenc}
  \usepackage{textcomp} 
\else 
  \usepackage{unicode-math}
  \defaultfontfeatures{Scale=MatchLowercase}
  \defaultfontfeatures[\rmfamily]{Ligatures=TeX,Scale=1}
\fi
\usepackage{lmodern}
\ifPDFTeX\else  
\fi
\IfFileExists{upquote.sty}{\usepackage{upquote}}{}
\IfFileExists{microtype.sty}{
  \usepackage[]{microtype}
  \UseMicrotypeSet[protrusion]{basicmath} 
}{}
\makeatletter
\@ifundefined{KOMAClassName}{
  \IfFileExists{parskip.sty}{%
    \usepackage{parskip}
  }{
    \setlength{\parindent}{0pt}
    \setlength{\parskip}{6pt plus 2pt minus 1pt}}
}{
  \KOMAoptions{parskip=half}}
\makeatother
\usepackage{xcolor}
\makeatletter
\ifx\paragraph\undefined\else
  \let\oldparagraph\paragraph
  \renewcommand{\paragraph}{
    \@ifstar
      \xxxParagraphStar
      \xxxParagraphNoStar
  }
  \newcommand{\xxxParagraphStar}[1]{\oldparagraph*{#1}\mbox{}}
  \newcommand{\xxxParagraphNoStar}[1]{\oldparagraph{#1}\mbox{}}
\fi
\ifx\subparagraph\undefined\else
  \let\oldsubparagraph\subparagraph
  \renewcommand{\subparagraph}{
    \@ifstar
      \xxxSubParagraphStar
      \xxxSubParagraphNoStar
  }
  \newcommand{\xxxSubParagraphStar}[1]{\oldsubparagraph*{#1}\mbox{}}
  \newcommand{\xxxSubParagraphNoStar}[1]{\oldsubparagraph{#1}\mbox{}}
\fi
\makeatother

\usepackage{longtable,booktabs,array}
\usepackage{calc} 
\usepackage{etoolbox}
\makeatletter
\patchcmd\longtable{\par}{\if@noskipsec\mbox{}\fi\par}{}{}
\makeatother
\IfFileExists{footnotehyper.sty}{\usepackage{footnotehyper}}{\usepackage{footnote}}
\makesavenoteenv{longtable}
\usepackage{graphicx}
\makeatletter
\def\maxwidth{\ifdim\Gin@nat@width>\linewidth\linewidth\else\Gin@nat@width\fi}
\def\maxheight{\ifdim\Gin@nat@height>\textheight\textheight\else\Gin@nat@height\fi}
\makeatother
\setkeys{Gin}{width=\maxwidth,height=\maxheight,keepaspectratio}
\makeatletter
\def\fps@figure{htbp}
\makeatother

\makeatletter
\@ifpackageloaded{caption}{}{\usepackage{caption}}
\AtBeginDocument{%
\ifdefined\contentsname
  \renewcommand*\contentsname{Table of contents}
\else
  \newcommand\contentsname{Table of contents}
\fi
\ifdefined\listfigurename
  \renewcommand*\listfigurename{List of Figures}
\else
  \newcommand\listfigurename{List of Figures}
\fi
\ifdefined\listtablename
  \renewcommand*\listtablename{List of Tables}
\else
  \newcommand\listtablename{List of Tables}
\fi
\ifdefined\figurename
  \renewcommand*\figurename{Figure}
\else
  \newcommand\figurename{Figure}
\fi
\ifdefined\tablename
  \renewcommand*\tablename{Table}
\else
  \newcommand\tablename{Table}
\fi
}
\@ifpackageloaded{float}{}{\usepackage{float}}
\floatstyle{ruled}
\@ifundefined{c@chapter}{\newfloat{codelisting}{h}{lop}}{\newfloat{codelisting}{h}{lop}[chapter]}
\floatname{codelisting}{Listing}

\makeatother
\makeatletter
\@ifpackageloaded{caption}{}{\usepackage{caption}}
\@ifpackageloaded{subcaption}{}{\usepackage{subcaption}}
\makeatother

\ifLuaTeX
  \usepackage{selnolig}  
\fi
\usepackage[]{natbib}
\usepackage{bookmark}

\IfFileExists{xurl.sty}{\usepackage{xurl}}{} 
\hypersetup{
  pdftitle={Stable Multivariate Functional Time Series Prediction for Major Geomagnetic Indices},
  pdfauthor={Yian Yu; Shasha Zou; Tuija Pulkkinen; Yang Chen},
  pdfkeywords={Overlapping rolling-window; conformal prediction; Geomagnetic indices forecasting},
  colorlinks=true,
  linkcolor={blue},
  filecolor={Maroon},
  citecolor={Blue},
  urlcolor={Blue},
  pdfcreator={LaTeX via pandoc}}

\newcommand{\anon}{1}

\usepackage{xr}
\begin{document}

\def\spacingset#1{\renewcommand{\baselinestretch}%
{#1}\small\normalsize} \spacingset{1}


\if1\anon
{
  \title{\bf Stable Multivariate Functional Time Series Prediction for Major Geomagnetic Indices}
\author[1]{Yian Yu}
\author[2]{Shasha Zou}
\author[2]{Tuija Pulkkinen}
\author[1]{Yang Chen\thanks{Corresponding author: ychenang@umich.edu}}
  \affil[1]{Department of Statistics, University of Michigan, Ann Arbor, MI, USA}
\affil[2]{Department of Climate and Space Sciences and Engineering, University of Michigan, Ann Arbor, MI, USA}
  \maketitle
} \fi

\if0\anon
{
  \bigskip
  \bigskip
  \bigskip
  \begin{center}
    {\LARGE\bf Stable Multivariate Functional Time Series\\[0.3em] Prediction for Major Geomagnetic Indices}
\end{center}
  \medskip
} \fi

\bigskip
\begin{abstract}
High-resolution scientific data, such as geomagnetic index streams, often exhibit complex temporal dependencies that can be modeled through functional data analysis. Conventional functional time series (FTS) methods typically partition continuous processes into non-overlapping segments, which artificially fragments temporal continuity and can limit estimation efficiency and stability. This is particularly evident in geomagnetic time series prediction due to their noisy, sudden, large-scale changes. This study presents a robust multivariate FTS forecasting framework for multi-dimensional time series with inter-series correlations and the existence of exogenous predictors. We introduce an overlapping rolling-window scheme that preserves temporal coherence and mitigates boundary information loss, thereby enriching the effective sample size for a more efficient and stable estimation. We integrate functional principal component analysis for dimension reduction with a vector autoregressive model with exogenous inputs to capture latent dynamics across correlated series. We also construct computationally efficient conformal prediction intervals for uncertainty quantification. The framework is motivated by and applied to the simultaneous forecasting of five critical geomagnetic indices, Kp, Dst, SYM-H, SME, SMR, using solar wind parameters as predictors. Empirical results show that this approach outperforms state-of-the-art machine learning baselines, extends forecast horizons to 6–24 hours, and provides calibrated uncertainty bounds.
\end{abstract}

\noindent%
\vfill

\newpage
\spacingset{1.45} 

\section{Introduction}\label{sec::Intro}
Advances in sensor technology and automated data acquisition have facilitated the collection of high-resolution observations in many scientific domains. The field of space weather forecasting has witnessed a strong boom of physics-informed data-driven approaches over the past two decades due to the availability of large volumes of streaming satellite data collected at high spatial and temporal resolutions \citep{Camporeale2019MLSW,chen2024solar}. With constant observations of solar activity and near-Earth magnetosphere, ionosphere, and thermosphere conditions, scientists are faced with a gigantic volume of data that can potentially improve forecasting accuracy, lead time, and reliability for many different space weather events. Among them, geomagnetic activity (created by solar and solar wind variations, such as solar storm) is monitored through multiple indices that summarize different aspects of the coupled magnetosphere--ionosphere system and are recorded continuously over time, which are closely associated with risks to power grids, satellite navigation and communication system, and satellites electronics \citep{sokolova2021geomagnetic,xue2024space}. Together with near-Earth solar-wind measurements, these observations form a collection of strongly dependent and heterogeneous time series with distinct temporal resolutions, dynamic ranges, and storm-time variability, as illustrated in Figure \ref{fig:IndicesFeatures}.  The figure shows strong temporal dependence within individual series, including persistence, abrupt storm-time transitions, extended recovery behavior, and cross-variable co-variation during active intervals. It also illustrates substantial heterogeneity across variables in scale, smoothness, intermittency, and dynamic range: some series evolve gradually, whereas others are characterized by sharp bursts or highly irregular fluctuations. 

Although several variables exhibit similar patterns during periods of disturbance, forecasting remains difficult because only information available up to the forecast time can be used. The model must therefore learn lagged and evolving dependencies rather than rely on contemporaneous alignment, which is visible only in retrospect. Existing data-driven studies on geomagnetic index forecasting frame the time series prediction problem as a classical machine learning task and are largely designed for individual index prediction, with forecasting skill often limited to short lead times of roughly 1--3 hours \citep{chandorkar2017probabilistic,laperre2020dynamic,park2021operational,Iong2020SYMH,hu2023multi,iong2024sparse,chen2025geodgp}. Despite promising root mean square error (RMSE) results, such works either fail to capture longer-range temporal dependency or are computationally expensive given the temporal dependencies; thus, this inaccurate characterization of the data results in less reliable out-of-sample forecasts. These limitations motivate a functional perspective, in which short temporal segments are treated as curves, and their sequential evolution is modeled over time.

\begin{figure}[ht]
    \centering
\includegraphics[width=1\linewidth]{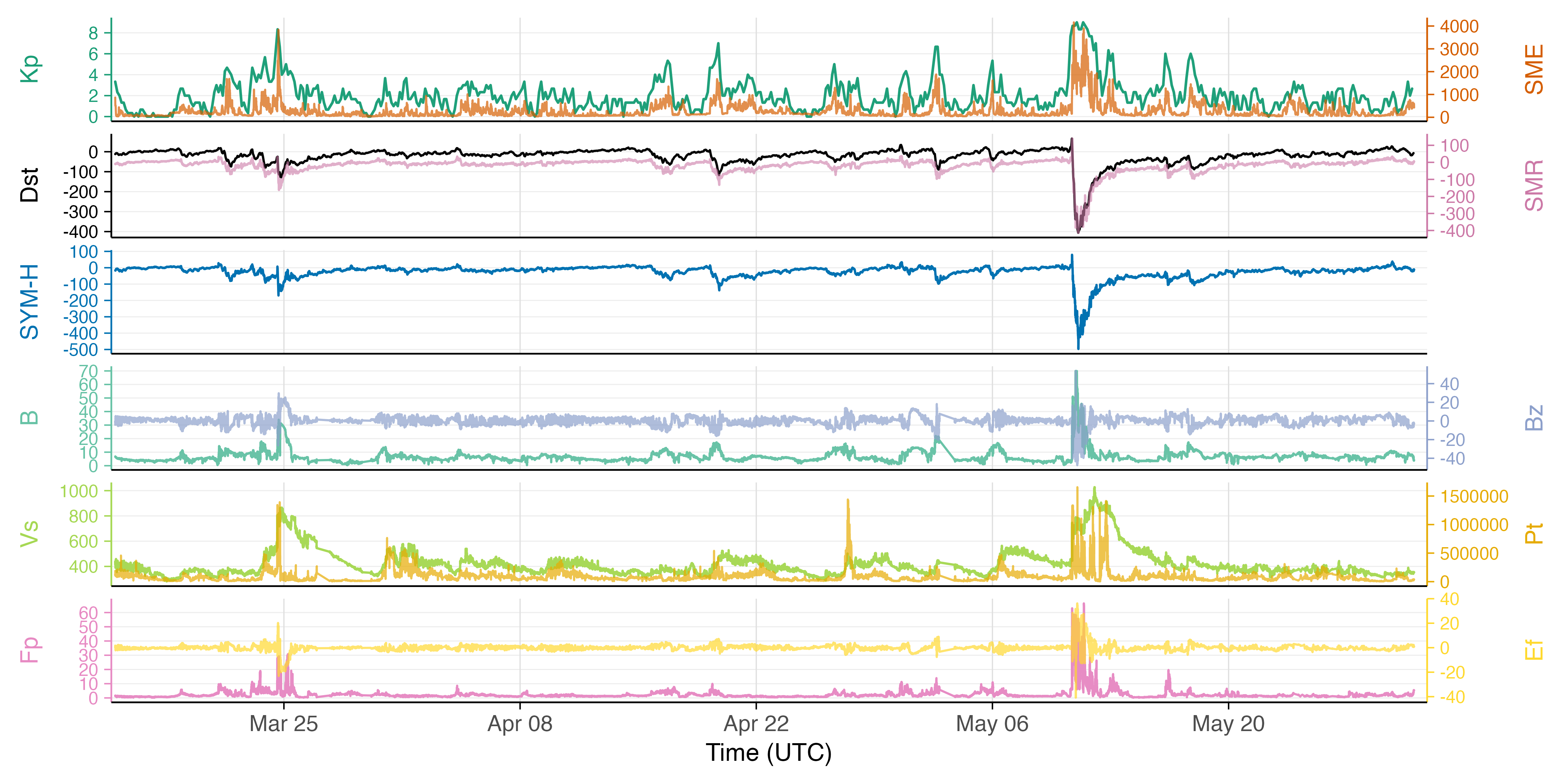}
    \caption{Representative time series of the geomagnetic indices Kp, Dst, SYM-H, SME, and SMR together with the near-Earth solar-wind and interplanetary magnetic field (IMF) variables B, Bz, Vs, Pt, Fp, and Ef over a selected interval in 2024.}
    \label{fig:IndicesFeatures}
\end{figure}

\subsection{Functional Time Series: Brief Literature}
Functional data analysis (FDA) treats individual observations as realizations of smooth functions, such as curves, images, or surfaces, rather than finite-dimensional vectors \citep{ramsay2005functional}. Typical examples include growth curves \citep{Erbas2007AgeFDA}, environment pollution concentration \citep{Aue2015SFTS}, electricity consumption \citep{Antoniadis2006FWK,shang2013functional}, and movement trajectories \citep{Yu202Baye,gertheiss2024functional}. 
Functional time series (FTS) arise when such functional observations are collected sequentially over time and are modeled as a stochastic process. 

\citet{hyndman2007robust} introduced a foundational FTS forecasting framework via functional principal component analysis (FPCA), where principal component (PC) scores are extracted from curves, forecasted in a lower-dimensional space, and then used to reconstruct future trajectories. This framework has motivated extensions such as robust variants for handling outliers \citep{shang2019robust} and functional ARMAX models that incorporate moving-average terms \citep{gonzalez2017forecasting}. Other decomposition-based methodologies include singular spectrum analysis (SSA), adapted to functional \citep{Haghbin2021FSSA,Trinka2023FSSA} and multivariate FTS contexts \citep{trinka2022multivariate}, alongside similarity-based methods employing weighted-similarity \citep{Antoniadis2006FKNN} or pattern sequence matching \citep{Bokde2017FPSF}. Nonetheless, scalable multivariate FTS with exogenous functional covariates remains challenging, especially when the observed series are noisy, heterogeneous, and strongly dependent. 

\subsection{Our Contribution: OLFTSA}
We propose a new multivariate FTS forecasting framework with exogenous regressors for high-resolution, heterogeneous, and strongly dependent geomagnetic index data observed together with solar-wind drivers. Our approach extends \citet{hyndman2007robust} by applying FPCA separately to each smoothed functional variable, rotating the stacked score matrix via orthogonal varimax for interpretability and parsimony \citet{kaiser1958varimax}, and model the resulting scores jointly through a vector autoregressive framework with exogenous inputs. This yields an estimation and forecasting procedure that is robust to heterogeneous resolutions, varied noise levels, and inconsistent fluctuation magnitudes across series. We refer to the proposed method as \textbf{O}ver\textbf{l}apping \textbf{F}unctional \textbf{T}ime \textbf{S}eries \textbf{A}nalysis with conformal prediction, denoted as \textbf{OLFTSA}, throughout the rest of the paper. 

Another distinctive feature of OLFTSA is the use of overlapping rolling windows to construct functional segments from the original time series; see Figure \ref{fig:model}. In contrast to conventional FTS approaches that partition continuous-time data into non-overlapping segments (e.g., by day, week, or year), this design preserves temporal continuity across adjacent segments, avoid losing transitional behavior at window boundaries, and enriches the effective sample size. Simulation studies in Section \ref{sec::Sim} show improved forecasting performance relative to non-overlapping constructions. This overlapping structure also facilitates the uncertainty quantification procedure developed in Section \ref{sec:UQ}, providing distribution-free validity under modest exchangeability conditions with tractable computation. We apply OLFTSA to the joint forecasting of multiple geomagnetic indices using solar-wind drivers and show that it delivers competitive 6--24 hour forecasts in real-data experiments, extending beyond the roughly 3-hour horizons typical of most existing single-index models constrained by the decay and complexity of solar wind forcing \citep{Khabarova2024}.

The rest of the manuscript is organized as follows. Section \ref{sec::Data} gives the construction of functional data with geomagnetic indices as responses and solar-wind parameters as predictors. Section \ref{sec:Method} presents the proposed model. Section \ref{sec::FTS} outlines the FTS methodological foundations.  Section \ref{sec::Sim} and Section \ref{sec:app} present simulations studies and real-data applications with benchmarks, respectively. Section \ref{sec:Disc} is the conclusion. Additional tables, figures, and theoretical proofs are provided in the Online Supplementary Materials.

\section{Geomagnetic Index Prediction Data}\label{sec::Data}

This study considers the joint forecasting of multiple geomagnetic indices using near-Earth solar wind and interplanetary magnetic field (IMF) measurements as predictors. The target indices, Kp, Dst, SYM-H, SME, and SMR, are crucial to space weather operations, serving as proxies for distinct magnetospheric current systems across varying temporal scales. Despite sharing a common primary driver in solar wind-magnetosphere interaction \citep{Newell2007SP}, these indices exhibit substantial heterogeneity in both sampling cadence and dynamic range. 
Consequently, 
studies in literature typically treat them as separate tasks \citep{Camporeale2019MLSW}. These indices exhibit non-negligible within-index and cross-index correlations, as shown in Figure \ref{fig:LagCorr_target} in the Online Supplementary Materials, which motivates a joint model. 

The target variables are five widely used geomagnetic indices: Kp from GFZ Potsdam (\url{https://kp.gfz.de/}), Dst and SYM-H from WDC Kyoto (\url{https://wdc.kugi.kyoto-u.ac.jp/}), and SME and SMR from SuperMAG (\url{https://supermag.jhuapl.edu/}). Together, these indices characterize complementary aspects of geomagnetic activity across different temporal scales and current systems. Specifically, Kp is a planetary index of global geomagnetic activity derived from mid-latitude ground magnetometers and reported every 3 hours. Dst is the disturbance storm-time index, which measures the intensity of the ring current using low-latitude magnetic observations at hourly cadence. SYM-H is a 1-minute high-temporal-resolution analogue of Dst, allowing finer characterization of rapid storm-time variations. SME is a high-latitude auroral electrojet index that quantifies auroral-zone horizontal ionospheric current activity and is particularly informative for substorm dynamics. SMR is a SuperMAG ring-current index analogous to Dst and SYM-H but constructed from a broader station network, providing improved spatial coverage.

The response set spans 1995--2025 and combines variables observed at heterogeneous cadences, ranging from 3-hourly Kp to hourly Dst and minute-level SYM-H, SME, and SMR. These indices also differ substantially in scale and dynamical behavior; for example, Kp ranges from 0 to 9, whereas SME may exceed 6000 nT during major geomagnetic disturbances. The high-frequency variables also exhibit short-scale variability and occasional redundancy at fine temporal resolution, motivating preprocessing before functional reconstruction. To reduce noise and eliminate near-duplicate adjacent observations at very fine resolution, we computed 5-minute averages for SYM-H, SME, and SMR \citep{Iong2020SYMH}.

The predictor set comprises six solar wind and IMF variables commonly used in geomagnetic index forecasting. 
Emphasis is placed on the southward IMF Bz, which governs dayside magnetic reconnection and energy input into the magnetosphere. Additionally, derived quantities such as the plasma flow dynamic pressure and solar wind motional electric field are included to capture known nonlinearities in solar wind magnetosphere coupling. More precisely, the predictor consists of solar wind and IMF parameters from NASA's OMNIWeb database (\url{https://omniweb.gsfc.nasa.gov/}), which is widely used in geomagnetic forecasting studies; see, for example, \citet{angeo2008GeoIndexOmni,Ayala2016SW,Wang2023AML,Iong2020SYMH,Nair2023SW,polozov2023sme}. In this study, we focus on a subset of predictors selected by examining cross-covariance with the target indices at multiple time lags; see Figure \ref{fig:LagCorr} in the Online Supplementary Materials. The selected variables are total IMF magnitude (B in nT), IMF component in the north-south direction in Geocentric Solar Magnetospheric (GSM) coordinates (Bz in nT; critical for magnetic reconnection), plasma flow speed (Vs in Km/s), proton temperature (Pt in K), flow dynamic pressure (Fp in nPa), and motional electric field (Ef in mV/m).

The OMNI predictors are available over 1995--2025 at high temporal resolution, including 1-minute products, but contain occasional missing values and are not directly aligned with the heterogeneous cadences of the response variables. To standardize the temporal grid and reduce preprocessing complexity, we linearly interpolated isolated missing values and resampled all OMNI predictors to 1-hour resolution. A summary table of the response and predictor variables is shown in Table \ref{tab:summary_var} in the Online Supplementary Materials.

\section{Model Architecture and Validation}
\label{sec:Method}

\subsection{Functional PCA (FPCA)}\label{sec::FTS}

We assume that the observations arise from an underlying smooth FTS $\{\mathrm{y}_d(\cdot),d\in\mathbb{N}\}$,
\[
y_d(t_i)=\mathrm{y}_d(t_i)+u_{d,i},\ i=1,\dots,n,\ d=1,\dots,D,
\]
where $\{t_i\}_{i=1}^n\subset\mathcal{T}$ are the sampling points at which the trajectories are measured, and $u_{d,i}$ are independent and identically distributed (i.i.d.) measurement errors with mean zero and variance $\sigma_u^2$. Throughout, the upright symbol $\mathrm{y}_d$ denotes the underlying curve and the italic symbol $y_d$ its noisy observations.

For discretely observed trajectories contaminated by noise, smoothing is used to recover the latent functional structure before estimating the mean and covariance functions. Each patch is presmoothed by a penalized B-spline estimator $\tilde{\mathrm{y}}_d(\cdot)$ of fixed spline order $q\ge r+1$ on a fixed knot sequence with $o(n)$ knots, using an $r$-th order difference penalty on the B-spline coefficients, with smoothing parameter $h_D$ chosen by generalized cross-validation. For an underlying curve with $r$ bounded derivatives, this estimator is asymptotically equivalent to a local polynomial smoother of degree $r-1$ \citep{ClaeskensSpline,Fanlocalpolyfitting1995}, and therefore attains the standard nonparametric rate $\E\{\tilde{\mathrm{y}}_d(t)-\mathrm{y}_d(t)\}^2=O\{h_D^{2r}+(nh_D)^{-1}\},\ t\in\mathcal{T},$
with $h_D\asymp n^{-1/(2r+1)}$. For sparse designs, where this curve-by-curve rate is not attainable, the mean and covariance can instead be estimated by pooling information across curves \citep{Yao01062005}.

Under weak stationarity, the lag-0 autocovariance operator $C$ of $\{\mathrm{y}_d\}_{d=1}^D$ is a nonnegative, self-adjoint Hilbert--Schmidt operator on $\mathcal{H}=L^2(\mathcal{T})$. By Mercer's theorem it admits the spectral decomposition $C(\cdot)=\sum_{k=1}^\infty\lambda_k\langle\phi_k,\cdot\rangle\phi_k$, where $\lambda_1\ge\lambda_2\ge\cdots\ge0$ are the eigenvalues and $\{\phi_k\}_{k\ge1}$ are the corresponding orthonormal eigenfunctions, $C(\phi_k)=\lambda_k\phi_k$ and $\|\phi_k\|=1$ \citep{Ash1975SP,hsing2015theoretical}. This yields the Karhunen--Lo\`eve expansion \citep{gohberg2012basic} $\mathrm{y}_d(t)=\mu(t)+\sum_{k=1}^\infty\xi_{d,k}\phi_k(t)$, $\xi_{d,k}=\langle\mathrm{y}_d-\mu,\phi_k\rangle$, $t\in\mathcal{T}$, where $\mu$ is the mean function. The population scores $\xi_{d,k}$ are uncorrelated random variables with mean zero and variance $\lambda_k$. Truncating the expansion at $K$ terms yields a finite-dimensional approximation in which the leading components often capture most of the variation in the sample paths \citep{RamsaySilverman2005,horvath2012inference}. Unlike spline, Fourier or wavelet expansions, FPCA uses the eigenfunctions of the covariance operator as a data-adaptive basis, which achieves effective dimension reduction with relatively few components \citep{horvath2012inference}.

The mean function and covariance kernel are estimated from the presmoothed curves by
\[
\hat\mu(t)=\frac1D\sum_{d=1}^D\tilde{\mathrm{y}}_d(t),\
\hat C(t,s)=\frac1D\sum_{d=1}^D\{\tilde{\mathrm{y}}_d(t)-\hat\mu(t)\}\{\tilde{\mathrm{y}}_d(s)-\hat\mu(s)\},\ t,s\in\mathcal{T},
\]
and the estimated eigenpairs $\{(\hat\lambda_k,\hat\phi_k)\}_{k=1}^K$ are obtained from the spectral decomposition of $\hat C$. The estimated scores are $\tilde\xi_{d,k}=\langle\tilde{\mathrm{y}}_d-\hat\mu,\hat\phi_k\rangle$, $k=1,\dots,K,$ $ d=1,\dots,D.$ Under standard regularity assumptions for densely observed presmoothed functional data \citep{bosq2000linear}, and provided that $K$ is fixed or grows slowly relative to $D$, $\hat\lambda_k$ and $\hat\phi_k$ are consistent at rate $D^{-1/2}$; Proposition \ref{prop:1} gives the refined rate that additionally accounts for discretization error, presmoothing bias and dependence. The measurement error variance is estimated by $\hat\sigma_u^2=\frac{1}{Dn}\sum_{d=1}^D\sum_{i=1}^n\{y_d(t_i)-\tilde{\mathrm{y}}_d(t_i)\}^2.$ The choice of $K$ remains application dependent. Common approaches include scree plots and the cumulative proportion of explained variance \citep{jolliffe2002principal}. Implementations are available in the R packages \texttt{fda} and \texttt{refund} \citep{RamsaySilverman2005,Goldsmith2013rfund}.

\subsection{OLFTSA Model Description}\label{ssec:Model}
Building on the functional representation of geomagnetic indices and solar wind drivers, we propose a framework that combines patch-wise functional preprocessing with an overlapping rolling-window scheme, FPCA-based dimension reduction and multivariate time series forecasting with exogenous inputs to produce multi-step forecasts of future index curves (patches), as summarized in Figure \ref{fig:model}. The framework accommodates heterogeneous sampling cadences and fluctuation magnitudes across variables by operating on smoothed functional patches and forecasting in a low-dimensional score space.

\paragraph*{Patch construction.} Let $\{y_p(t_i),i\in\mathbb{N}\}$ denote the observation stream of variable $p\in\mathcal{P}\cup\mathcal{Q}$, where $\mathcal{P}=\{\mathrm{Kp},\mathrm{Dst},\mathrm{SYM\text{-}H},\mathrm{SME},\mathrm{SMR}\}$ is the set of target variables and $\mathcal{Q}=\{\text{B},\text{Bz}, \text{Vs}, \text{Pt}, \text{Fp},\text{Ef}\}$ is the set of solar wind and IMF input variables. The stream is reshaped into fixed-duration patches, such as 48- or 72-hour windows, indexed by $d=1,\dots,D$. The $d$-th patch of variable $p$ is $y_{p,d}(t_{p,1}),\dots,y_{p,d}(t_{p,n_p})$, where $\{t_{p,i}\}_{i=1}^{n_p}\subset\mathcal{T}$ is the within-patch time grid, which may differ across variables. Consecutive patches are shifted by the forecast horizon, so that the first $n_p^\prime$ grid points of patch $d+1$ are observed at the forecast origin and the last $n_p-n_p^\prime$ grid points form the forecast block; the overlap ratio is $\gamma=n_p^\prime/n_p$. Empirical autocorrelation decays substantially within about one day for all variables (Figure \ref{fig:LagCorr_target} of the Online Supplementary Materials), which motivates overlapping rolling windows that preserve short-range dependence across adjacent patches.

\begin{figure}[hpbt!]
  \centering
\resizebox{\textwidth}{!}{\input{Model-figure}}
\caption{Schematic illustration of the proposed OLFTSA framework. The embedded plots use Kp as an illustrative example of the overlapping rolling-window construction and the associated FPCA quantities. $\mathcal{P}=\{\text{Kp},\text{Dst},\text{SYM-H},\text{SME},\text{SMR}\}$ is the set of target geomagnetic indices and $\mathcal{Q}=\{\text{B},\text{Bz},\text{Vs},\text{Pt},\text{Fp},\text{Ef}\}$ is the set of input solar wind and IMF variables. FPCA denotes functional principal component analysis, and MTSF denotes multivariate time series forecasting.}
  \label{fig:model}
\end{figure}

\paragraph*{Functional representation.} Each patch is treated as a discretely observed realization of an underlying curve $\mathrm{y}_{p,d}(\cdot)$ on $\mathcal{T}$ and presmoothed by the penalized B-spline estimator of Section \ref{sec::FTS}, which gives $\tilde{\mathrm{y}}_{p,d}(\cdot)$. Applying the FPCA of Section \ref{sec::FTS} separately to each variable $p$ yields the estimated mean function $\hat\mu_p$, eigenfunctions $\{\hat\phi_{p,k}\}_{k\ge1}$ and scores $\tilde\xi_{p,d,k}=\langle\tilde{\mathrm{y}}_{p,d}-\hat\mu_p,\hat\phi_{p,k}\rangle$, which estimate their population counterparts $\mu_p$, $\phi_{p,k}$ and $\xi_{p,d,k}$. Truncation to $K_p\le n_p$ components, with $K_p$ chosen by a pre-specified proportion of explained variance such as 0.99, gives
\[
\tilde{\mathrm{y}}_{p,d}(t)\approx\hat\mu_p(t)+\sum_{k=1}^{K_p}\tilde\xi_{p,d,k}\hat\phi_{p,k}(t),\qquad t\in\mathcal{T}.
\]
For each variable separately, the $D\times K_p$ score matrix is rotated by its own orthogonal varimax transform $R_p\in\Re^{K_p\times K_p}$ \citep{kaiser1958varimax}, and the eigenfunctions $(\hat\phi_{p,1},\dots,\hat\phi_{p,K_p})$ are rotated by the same $R_p$, so that the resulting scores load more sparsely on interpretable, time-localized patterns. Because $R_p$ is orthogonal, the rotation is a norm-preserving change of basis that leaves the reconstructed curves and the estimation theory of Proposition \ref{prop:1} unchanged; for notational simplicity we use the same symbols for the rotated scores and eigenfunctions.

\paragraph*{Score-space forecasting.} Let $\tilde{\ve\xi}_{p,d}=(\tilde\xi_{p,d,1},\dots,\tilde\xi_{p,d,K_p})^\top$ denote the estimated score vector of variable $p$ at patch $d$. Stacking the score vectors of the targets and the inputs gives the multivariate score time series $\{\tilde{\ve\xi}_{\mathcal{P}\cup\mathcal{Q},d}\}_{d=1}^D$ of dimension $K_{\text{tot}}=\sum_{p\in\mathcal{P}\cup\mathcal{Q}}K_p$, and $\tilde{\ve\xi}_{\mathcal{P},d}$ collects the $K_{\mathcal{P}}=\sum_{p\in\mathcal{P}}K_p$ target scores. We forecast the target scores of the next patch with a vector autoregression with exogenous inputs of order one, VARX(1):
\begin{equation}\label{eq:varx}
\hat{\ve\xi}_{\mathcal{P},d+1}=\hat{\ve B}\,\tilde{\ve\xi}_{\mathcal{P}\cup\mathcal{Q},d}+\hat{\ve c},
\end{equation}
where $(\hat{\ve B},\hat{\ve c})$ is the ridge-regularized least-squares estimator
\[
(\hat{\ve B},\hat{\ve c})=\arg\min_{\ve B,\ve c}\ \sum_{d=1}^{D-1}\big\|\tilde{\ve\xi}_{\mathcal{P},d+1}-\ve B\,\tilde{\ve\xi}_{\mathcal{P}\cup\mathcal{Q},d}-\ve c\big\|^2+\varpi\big(\|\ve B\|_{\text{F}}^2+\|\ve c\|^2\big),
\]
computed on standardized scores, with the penalty $\varpi>0$ chosen by rolling-origin cross-validation on the training set \citep{Basu2015VAR,Ballarin2025RidgeVAR}. The population counterpart of \eqref{eq:varx}, with coefficients $(\ve B,\ve c)$ and population scores, is given in Condition \ref{con:4}. The lag order $\ell=1$ was selected by minimizing the rolling-origin cross-validated one-step-ahead RMSE over $\ell\in\{1,2,3,4\}$ on the training set.

\paragraph*{Back-projection.} The forecast score vectors are mapped back to the functional domain through the retained rotated eigenfunctions, $\hat y_{p,d+1}(t)=\hat\mu_p(t)+\sum_{k=1}^{K_p}\hat\xi_{p,d+1,k}\hat\phi_{p,k}(t),$ $p\in\mathcal{P},$
and evaluated at the grid points $t_{p,n_p^\prime+1},\dots,t_{p,n_p}$ of the forecast block, which gives the forecast vector $\hvesub{y}{p,d+1}\in\Re^{n_p-n_p^\prime}$ used in Section \ref{sec:UQ}. Although the patches are constructed from overlapping windows, the overlap mainly affects the temporal dependence among curves rather than the FPCA itself. Under the eigenvalue spacing, smoothing consistency and weak dependence Conditions \ref{con:1}--\ref{con:3} of Proposition \ref{prop:1}, the estimated mean function, covariance kernel and leading eigencomponents remain consistent \citep{bosq2000linear,horvath2012inference,Hormann2010WeekdepFD}, and Proposition \ref{prop:1} gives the resulting one-step-ahead forecast error rate under Condition \ref{con:4}.

\subsection{Uncertainty Quantification for OLFTSA}\label{sec:UQ}
We derive the uncertainty quantification procedure for  OLFTSA, adapting the sequential conformal prediction framework of \citet{xu2024CPMTS-ellipsoidal}. 
Recall that we constructed FTS using an overlapping rolling-window scheme, as illustrated in Figure \ref{fig:model}. 
For notational simplicity, we suppress the subscript $p$ and define the actual forecast vector for $d$-th patch as $\hvesub{y}{d}=\left(\hat{y}_{d}(t_{n^\prime+1}), \cdots, \hat{y}_{d}(t_{n}) \right)^\top$, and the corresponding true values as $\vesub{y}{d}=\left(y_{d}(t_{n^\prime+1}), \cdots, y_{d}(t_{n}) \right)^\top$.
We aim to construct a prediction region $\mathcal{R}_{d}^\alpha$ for each patch $d \ge D+1$, conditional on past observations and the current timestamps $\vesub{t}{d}=\ve{t}=\{t_i\}_{i=n^\prime +1}^{n}$, for a specified significance level $\alpha \in (0,1)$, with coverage guarantees: a weaker marginal coverage and/or a a stronger conditional coverage \citep{gibbs2025conformal}. 
Under coverage guarantees, regions with small volume provide more informative uncertainty quantification. 

Assume the forecasting model is trained. We use the first $D$ patches as a calibration set, we compute residual vectors, $\hvesub{\varepsilon}{d} = \vesub{y}{d} - \hvesub{y}{d} \in \Re^{n-n'},$ $d=1,\cdots, D$, and their empirical mean $\bar{\ve{\varepsilon}} = \frac{1}{D}\sum_{d=1}^D \hvesub{\varepsilon}{d}$. 
Let the sample covariance be $\tilde{\Sigma}_{D} = \frac{1}{D}\sum_{d=1}^{D}(\hvesub{\varepsilon}{d}-\bar{\ve{\varepsilon}})(\hvesub{\varepsilon}{d}-\bar{\ve{\varepsilon}})^\top$. In moderate- to high-dimensional settings, $\tilde{\Sigma}_{D}$ can be ill-conditioned or, once $n-n^\prime$ approaches $D$, singular.

Entrywise thresholding alone, $\left[ \tilde{\sigma}_{j,k}  \mathbb{I}(|\tilde{\sigma}_{j,k}| > \tau) \right]_{1 \leq j,k \leq n-n^\prime},$ where $\tilde{\sigma}_{j,k}$ denotes the $(j,k)$-th entry in $\tilde{\Sigma}_{D}$ and $\tau > 0$ is a tuning parameter, reduces estimation variance but does not guarantee a positive-definite matrix, which can zero out enough off-diagonal mass to make the resulting matrix singular or indefinite, which would leave the nonconformity score below undefined \citep{KIM2023Cov}. We therefore add an explicit positive-definite projection step. Let $\left[ \tilde{\sigma}_{j,k}  \mathbb{I}(|\tilde{\sigma}_{j,k}| > \tau) \right]_{1 \leq j,k \leq n-n^\prime}$ with eigendecomposition $\left[ \tilde{\sigma}_{j,k}  \mathbb{I}(|\tilde{\sigma}_{j,k}| > \tau) \right]_{1 \leq j,k \leq n-n'}=\sum_{i=1}^{n-n^\prime} \kappa_{i}\vesub{v}{i}\vess{v}{i}{\top}$, and define, for a fixed constant $\iota > 0$, $\hat{\Sigma}_{D} = \sum_{i=1}^{n_p-n^\prime}\max(\kappa_{i}, \iota) \vesub{v}{i}\vess{v}{i}{\top}$. By construction and for every $D$, $\lambda_{\min}(\hat{\Sigma}_{D}) \geq \iota >0$, where $\lambda_{\min}(\cdot)$ denotes the smallest eigenvalue. So $\hat{\Sigma}_{D}$ is invertible with the spectral (operator) norm $\|\hat{\Sigma}_{D}^{-1} \|_{\text{op}}\leq \iota^{-1}$.

We use a Mahalanobis-type quadratic form as the nonconformity score. For any realization $\ve{y}$ for patch $d$, define \begin{equation}
     \hat{e}_{d}(\ve{y})=\left( \ve{y}-\hvesub{y}{d}-\bar{\ve{\varepsilon}} \right)^\top \hat{\Sigma}_{D}^{-1} \left(\ve{y}-\hvesub{y}{d}- \bar{\ve{\varepsilon}}\right) \in [0,\infty). \label{eq:score}
\end{equation} The realized calibration scores are $\hat{e}_{d}:= \hat{e}_{d}(\vesub{y}{d})$ for $d=1,\cdots, D$ with score sequence $\hve{e}=\{\hat{e}_{d}\}_{d=1}^D$. Because of the inherent dependency among the original data $\left\{\vesub{y}{d} \right\}_{d=1}^D$, the scores in $\hve{e}$ are not exchangeable. Instead of using the empirical $(1-\alpha)$-quantile of $\hve{e}$, we estimate an adaptive quantile of the score process using the sequential conformal quantile regression approach of \citet{xu2024CPMTS-ellipsoidal}. See also the quantile random forests construction in \citet{Xu23SPCI}. Concretely, let $\hat{Q}_{d}(\alpha)$ denote the fitted estimator of the $\alpha$-quantile of $\hat{e}_{d}$, optionally conditional on predictors such as $\vesub{t}{d}$ or other available covariates. We then define the ellipsoidal prediction region for patch $d \ge D+1$ of variable $p$ as \begin{equation} \label{eq:PredictionRegion}
     \mathcal{R}_{d}^\alpha= \left\{\ve{y}: \hat{Q}_{d}(\hat{\beta})\leq \hat{e}_{d}(\ve{y}) \leq  \hat{Q}_{d}(1-\alpha + \hat{\beta}) \right\},\end{equation} where
    \begin{align*}         \hat{\beta}=&\argmin_{\beta \in [0,\alpha]} \left\{ V\left(\hat{\Sigma}_{D} , \hat{Q}_{d}(1-\alpha+\beta) \right)  - V\left( \hat{\Sigma}_{D},\hat{Q}_{d}(\beta) \right)  \right\},
    \end{align*} $V(\hat{\Sigma}_{D},r)= \text{vol}\left\{ \ve{y} \in\Re^{n-n^\prime}: \vesup{y}{\top}\hat{\Sigma}_{D}^{-1}\ve{y}\leq r\right\}$, $r\geq 0,$ is the volume of the sub-level ellipsoid of the quadratic form defining $\hat{e}_d$, with the explicit form $$V(\hat{\Sigma}_{D},r)=|\hat{\Sigma}_{D}|^{1/2}\frac{r^{(n-n^\prime)/2}\pi^{(n-n^\prime)/2}}{\Gamma((n-n^\prime)/2+1)}.$$ The resulting set is an ellipsoid centered at $\hvesub{y}{d} + \bar{\ve{\varepsilon}}$ with shape determined by $\hat{\Sigma}_{D}^{-1}$. The full sequential calibration algorithm is given in \citet{xu2024CPMTS-ellipsoidal}.

We now establish an asymptotic coverage bound for the proposed conformal prediction region under dependent forecast errors. Let ${y}_{d}^{\ast}$ denote the oracle forecast of patch $d$, obtained by applying the population mean, eigenfunctions and VARX(1) coefficients to the population scores of patch $d-1$, such as ${y}_{d}^{\ast}(t)=\ve{\mu}(t)+\sum_{k\leq K} {\xi}_{d,k}^{\ast}\phi_d(t)$, $\vess{\xi}{\mathcal{P},d}{\ast}=\ve{B}\vesub{\xi}{\mathcal{P}\cup\mathcal{Q},d-1}+\ve{c}$, and let $\vesub{\varepsilon}{d}=\vesub{y}{d}-\ve{y}^\ast_d$ be the corresponding error vector at the $n-n^\prime$ forecast time points, with mean zero and covariance matrix $\Sigma$. Define the theoretical score $e_d=\vess{\varepsilon}{d}{\top}\Sigma^{-1}\vesub{\varepsilon}{d}$, and let $\hvesub{\varepsilon}{d}=\vesub{y}{d}-\hvesub{y}{d}-\bar{\ve\varepsilon}$ denote the centred residual that enters the score in \eqref{eq:score}. Relaxing the standard assumption that $\{\vesub{\varepsilon}{d}\}_{d=1}^{D+1}$ are i.i.d., we assume that the sequence is stationary and strong mixing. Under Assumptions \ref{Assum:error_noniid}-\ref{Assum:weakdepen&tail}, we bound the deviation of the coverage probability from its nominal level, accounting for estimation error in the noise covariance matrix, which is not considered in \citet{xu2024CPMTS-ellipsoidal}; see Section \ref{Appen:ProofThem1} of the Online Supplementary Materials. Throughout, the forecasting model is fitted on a training period preceding the $D$ calibration patches and held fixed over them. The notation used in the proofs is summarized in Table \ref{tab:notation} of the Online Supplementary Materials.

\begin{assumption}[Stationarity and Strong Mixing]\label{Assum:error_noniid} The error sequence $\{\vesub{\varepsilon}{d}\}_{d=1}^{D+1}$ is stationary and strong mixing (or $\alpha$-mixing) with mixing coefficients satisfying $\sum_{l>0} \alpha_{l} < M$ for some constant $M<\infty$, where $\alpha_{l} = \sup_{d \in \mathbb{N}} \sup_{\mathcal{A} \in \mathcal{F}_{1}^{d}, \mathcal{B} \in \mathcal{F}_{d+l}^{\infty}} |P(\mathcal{A} \cap \mathcal{B}) - P(\mathcal{A})P(\mathcal{B})|$ and $\mathcal{F}_{a}^{b}$ denotes the $\sigma$-algebra generated by $\{\vesub{\varepsilon}{a}, \cdots, \vesub{\varepsilon}{b}\}$.
\end{assumption}
Stationarity implies that for any $j \ge 1$, integers $d_1, \cdots, d_j$, and lag $l$, the joint distribution of $(\vesub{\varepsilon}{d_1}, \cdots, \vesub{\varepsilon}{d_j})$ is identical to that of $(\vesub{\varepsilon}{d_1+l}, \cdots, \vesub{\varepsilon}{d_j+l})$. The strong mixing condition requires that the coefficients $\alpha_{l}$ converge to zero as $l \to \infty$.

\begin{assumption}[Lipschitz Continuity]\label{Assum:CDFLip}
The cumulative distribution function of the true nonconformity score, denoted by $F_e(x)=P(e<x)$, is Lipschitz continuous with constant $L_{D+1} >0$, that is, $F_e(x_1)-F_e(x_2) \leq L_{D+1}|x_1-x_2|$ for all $x_1, x_2 \in \Re$.
\end{assumption}

\begin{assumption}[Prediction Error Bound]\label{Assum:error_sum}
There exists a sequence $\{\delta_D\}_{D\ge1}$ such that
(i) $\frac{1}{D}\sum_{d=1}^D\|\hvesub{\varepsilon}{d}-\vesub{\varepsilon}{d}\|^2\le\delta_D^2$ and $\|\hvesub{\varepsilon}{D+1}-\vesub{\varepsilon}{D+1}\|\le\delta_D$; and
(ii) $\E\|\hvesub{\varepsilon}{d}-\E\hvesub{\varepsilon}{d}-\vesub{\varepsilon}{d}\|^2\le\delta_D^2$ for $1\le d\le D+1$.
\end{assumption}
Part (i) bounds the realized prediction error and part (ii) its population counterpart, which controls the difference between the covariance of the residual process and $\Sigma$. For consistent estimators, $\delta_D\to0$ as $D\to\infty$.

\begin{assumption}[Covariance Regularity]\label{Assum:bound_cov}
There exists a constant $\iota > 0$ such that $\lambda_{\min}(\Sigma) \ge \iota$ and $\lambda_{\min}\left(\hat{\Sigma}_{D}\right) \ge \iota$, where $\lambda_{\min}(\cdot)$ denotes the smallest eigenvalue.
\end{assumption}
This condition ensures that both the true and estimated covariance matrices remain strictly positive definite, preventing singularity issues in the score computation.

\begin{assumption}[Weak Dependence and Moment Conditions]\label{Assum:weakdepen&tail}
The residual process admits the representation $\hvesub{\varepsilon}{d}=\ve{g}(\dots,\zeta_{d-1},\zeta_d)$ for i.i.d.\ innovations $\{\zeta_d\}$ and a measurable function $\ve g$. Let $\hvesub{\varepsilon}{d}^{\{0\}}$ be obtained by replacing $\zeta_0$ with an independent copy, and define the functional dependence measure \citep{Chen2013Covariance}
\[
\theta_{d,\omega}=\max_{n^\prime+1\le i\le n}\left(\E|\hat\varepsilon_{d,i}-\hat\varepsilon^{\{0\}}_{d,i}|^\omega\right)^{1/\omega},\qquad
\Theta_{j,\omega}=\sum_{d=j}^{\infty}\theta_{d,\omega}.
\]
There exist $\omega>4$ and positive constants $c_1$ and $c_2$ such that $\Theta_{j,\omega}\le c_2j^{-c_1}$ for all $j\ge1$. For tail behavior, there exist positive constants $c_3,c_4,c_5$ such that $\max_{n^\prime+1\le i\le n}\E|\hat\varepsilon_{d,i}|^\omega<c_3$, $\max_{1\le d\le D+1}\|\vesub{\varepsilon}{d}\|\le\sqrt{c_4\omega}$ almost surely, and $\Var(\|\vesub{\varepsilon}{d}\|^2)\le c_5\omega$.
\end{assumption}

Assumptions \ref{Assum:error_sum} and \ref{Assum:bound_cov} are stated in terms of the residual process $\hvesub{\varepsilon}{d}$. The following proposition gives interpretable conditions on the raw FTS that imply them when the presmoothing method is a penalized B-spline estimator and the forecasting map is a ridge-regularized VARX(1); the proof is provided in Section \ref{Appen:ProofProp1} of the Online Supplementary Materials.

\begin{prop}[Primitive sufficient conditions]\label{prop:1}
Suppose:

\begin{condition}[Eigenvalue spacing]\label{con:1}
The number of components $K$ is fixed, and the leading eigenvalues $\lambda_1>\cdots>\lambda_K$ of the lag-0 covariance operator $C$ satisfy $\min_{1\le k\le K}(\lambda_k-\lambda_{k+1})\ge c_\lambda>0$. In addition, the eigenfunctions are uniformly bounded, $\sup_{k\ge1}\|\phi_k\|_\infty:=C_\phi<\infty$.
\end{condition}

\begin{condition}[Smoothing consistency]\label{con:2}
The presmoothing estimator $\hat{\mathrm{y}}_d(t)$ of the underlying curve $\mathrm{y}_d(t)$ satisfies $\sup_{t\in\mathcal{T}}\E\{\hat{\mathrm{y}}_d(t)-\mathrm{y}_d(t)\}^2=O\{h_D^{2r}+(nh_D)^{-1}\}$, uniformly in $d$, for an $r$-times differentiable target curve and bandwidth $h_D\to0$.
\end{condition}
\begin{condition}[Weak dependence of the raw FTS]\label{con:3}
The raw patch process $\{y_d(\cdot),d\in\mathbb{N}\}$ is a causal Bernoulli shift with $\E\|y_d\|^{\omega}<\infty$ for some $\omega>4$ and functional dependence measure $\theta^{y}_{d,\omega}=(\E\|y_d-y_d^{\{0\}}\|^{\omega})^{1/\omega}$ satisfying $\sum_{d\ge j}\theta^{y}_{d,\omega}=O(j^{-c_1})$; this implies $L^\omega$-$m$-approximability \citep{Hormann2010WeekdepFD}. Moreover, $\{y_d\}$ is $\alpha$-mixing with summable coefficients \citep{bosq2000linear}.
\end{condition}

\begin{condition}[Forecasting map]\label{con:4}
The population FPC scores satisfy
$\vesub{\xi}{\mathcal{P},d+1}=\ve{B}\vesub{\xi}{\mathcal{P}\cup\mathcal{Q},d}+\ve{c}+\vesub{\eta}{d+1}$,
where $(\ve B,\ve c)$ are the population least-squares coefficients, which do not depend on $D$, so that $\E\vesub{\eta}{d+1}=\ve0$ and $\E(\vesub{\eta}{d+1}\vess{\xi}{\mathcal{P}\cup\mathcal{Q},d}{\top})=\ve0$. Partition $\ve B=[\ve{B}_{\mathcal{PP}},\ve{B}_{\mathcal{PQ}}]$ conformably with $\vesub{\xi}{\mathcal{P}\cup\mathcal{Q},d}=(\vesub{\xi}{\mathcal{P},d},\vesub{\xi}{\mathcal{Q},d})$, $K_{\mathcal P}=\sum_{p\in\mathcal P}K_p$. Assume $|\ve{I}_{K_{\mathcal P}}-\ve{B}_{\mathcal{PP}}z|\ne0$ for all complex $|z|\le1$ and $\|\ve B\|_{\text{op}}=\rho_B<\infty$. The ridge-regularized estimator $(\hat{\ve B},\hat{\ve c})$ of Section 3.2, computed from the estimated scores, satisfies
$\|\hat{\ve B}-\ve B\|_{\text{op}}+\|\hat{\ve c}-\ve c\|=O_p\{\sqrt{K_{\text{tot}}/D}+h_D^r+(nh_D)^{-1/2}\}$, $K_{\text{tot}}=\sum_{p\in\mathcal P\cup\mathcal Q}K_p$ \citep{Basu2015VAR,Ballarin2025RidgeVAR}.
\end{condition}

Then Assumption \ref{Assum:error_sum} holds in probability with
$\delta_D=h_D^r+(nh_D)^{-1/2}+\sqrt{K_{\text{tot}}/D}$,
that is, for every $\epsilon>0$ there exists $C_\epsilon$ such that part (i) holds with $C_\epsilon\delta_D$ with probability at least $1-\epsilon$, and part (ii) holds if the squared estimation errors are uniformly integrable. Moreover, the dependence part of Assumption \ref{Assum:weakdepen&tail} holds with $\Theta_{j,\omega}=O(j^{-c_1})$, the same polynomial rate as in Condition \ref{con:3}, and Assumption \ref{Assum:error_noniid} holds.
\end{prop}

The truncation remainder $\sum_{k>K}\xi_{d,k}\phi_k$ is part of $\vesub{\varepsilon}{d}$. It therefore enlarges $\Sigma$, and hence the prediction region, but does not affect the rate $\delta_D$ or the validity of the region.

\begin{theorem}\label{Thm:coverage}
Under Assumptions \ref{Assum:error_noniid}--\ref{Assum:weakdepen&tail},  for any calibration sample size $D$, confidence level $\alpha\in(0,1)$, any $\beta\in[0,\alpha]$ used to define $\mathcal{R}^\alpha_{D+1}$ in Equation \eqref{eq:PredictionRegion}, and any $\delta\in(0,1)$,
\[
\left|P\left(\vesub{y}{D+1}\in\mathcal{R}^\alpha_{D+1}\right)-(1-\alpha)\right|
\le\frac{14M^{1/3}(\log D)^{2/3}}{(2D)^{1/3}}+4\left(L_{D+1}G(D,\tau)+2L_{D+1}+1\right)G(D,\tau)+2\delta,
\]
where
\begin{equation}
\begin{aligned}
G(D,\tau)&=\left\{\frac{\delta_D^2}{\iota}+\frac{2\sqrt{c_4\omega}}{\iota}\,\delta_D+\frac{c_4\omega}{\iota^2}\left(\sqrt{\frac{C_1}{\delta}A(D,\tau)}+\delta_D^2+2\sqrt{c_4\omega}\,\delta_D\right)\right\}^{1/2},\\
A(D,\tau)&=\left(\sum_{j,k=1}^{n-n^\prime}\min(\tau^2,\sigma_{jk}^2)\right)+(n-n^\prime)^2\min\left(\frac1D,\frac{\tau^{2-\omega/2}}{D^{\omega/4}},H(\tau)\right),
\end{aligned}\label{Eq:G(D,tau)}
\end{equation}
and
\[
H(\tau)=
\begin{cases}
\dfrac{\tau^{2-\omega/2}}{D^{\omega/2-1}}+C_2\left(\dfrac1D+\tau^2\right)e^{-D\tau^2}, & \text{if } c_1>1/2-2/\omega,\\[2ex]
\dfrac{\tau^{2-\omega/2}(\log D)^{1+\omega/2}}{D^{\omega/2-1}}+C_2\left(\dfrac{(\log D)^2}{D}+\tau^2\right)e^{-\frac{D\tau^2}{(\log D)^2}}, & \text{if } c_1=1/2-2/\omega,\\[2ex]
\dfrac{\tau^{2-\omega/2}}{D^{\omega/2(c_1+1/2)}}+C_2\left(\dfrac{1}{D^{\tilde c_1}}+\tau^2\right)e^{-D^{\tilde c_1}\tau^2}, & \text{if } c_1<1/2-2/\omega,
\end{cases}
\]
with $\tilde c_1=(3+c_1\omega)/(1+\omega/2)$. Here, $C_1$ and $C_2$ are constants independent of $\tau$, $D$ and $\omega$.
\end{theorem}

Because the time stamps $\vesub{t}{D+1}$ are fixed by design, the probability in Theorem \ref{Thm:coverage} is taken over the joint distribution of the calibration and test patches.

\begin{remark}
If the threshold $\tau$ is of order $D^{-1/2}$, then $A(D,\tau)=O(D^{-1})$ \citep{Chen2013Covariance}, and, for $\delta_D\le1$, $G(D,\tau)=O(\delta_D^{1/2}+\delta^{-1/4}D^{-1/4})$. Choosing $\delta=D^{-1/5}$ to balance the last two terms of the bound gives
\[
\left|P\left(\vesub{y}{D+1}\in\mathcal{R}^\alpha_{D+1}\right)-(1-\alpha)\right|=O\left(D^{-1/5}+\delta_D^{1/2}\right),
\]
which, under the conditions of Proposition \ref{prop:1}, is $O\{D^{-1/5}+h_D^{r/2}+(nh_D)^{-1/4}\}$. The coverage error therefore vanishes as $D\to\infty$, $h_D\to0$ and $nh_D\to\infty$.
\end{remark}


\subsection{Simulation Studies}\label{sec::Sim}

In this section, we evaluate the proposed OLFTSA method for multivariate functional forecasting and uncertainty quantification. Synthetic data are generated following the process in \citet{Trinka2023FSSA}, adapted for continuous FTS under rolling-window segmentation. This enables a direct comparison with FSSA \citep{Trinka2023FSSA} using the residual bootstrap for uncertainty quantification. We evaluate forecast accuracy by the root mean square error (RMSE) and mean absolute error (MAE), and the validity of the conformal prediction sets by the empirical coverage probability (ECP) and mean interval width (MIW) at the nominal level $0.95$.

The univariate continuous underlying process $z(s)$ is simulated over a global time domain $s \in [0, T]$ with $T=100$, discretized at $N=\tilde{D}n$ equally spaced points $\{s_i\}_{i=1}^N$. Here, $n$ denotes the length of each segmented curve, so that a larger $n$ yields a finer temporal cadence, and $\tilde{D}$ is the number of curves generated per raw draw; a larger $\tilde{D}$ increases the total number of patches and allows the asymptotic coverage to be assessed.
The underlying process is decomposed as $z(s) = m(s) + h(s)$, where $m(s)$ is the mean trend given by
$m(s) = \nu s + \exp(\tilde{s}^2)\cos(2\pi \omega s) + \cos(4\pi \tilde{s})\sin(2\pi \omega s)$, with drift rate $\nu = 0.02$, frequency $\omega = 0.2$, and local daily cycle $\tilde{s}=s \pmod{1}$. The term $h(s)$ is a stochastic component constructed through a basis expansion with time-varying coefficients, $h(s) = \sum_{j=1}^J \zeta_j(s) \psi_j(\tilde{s})$, where $\{\psi_j(\cdot)\}_{j=1}^J$ are fixed Fourier basis functions on the local domain with $J=5$, which ensures a smooth functional structure. The latent coefficients $\ve{\zeta}(s) = (\zeta_1(s), \ldots, \zeta_J(s))^\top$ follow a vector autoregressive process of order $1$,
$\ve{\zeta}(s)=\ve{A} \ve{\zeta}(s-\Delta_{s})+\ve{\eta}(s)$, $\ve{\eta}(s) \sim N(0,0.05^2 \mathbf{I}_J)$,
where $\Delta_s$ is the discretization step. To produce strong day-to-day autocorrelation (approximately $0.8$), the transition matrix $\ve{A}\in \Re^{J \times J}$ is diagonal with decreasing diagonal entries. One sample trajectory is shown in Figure \ref{fig:OneSimu} in the Online Supplementary Materials.

To transform $z(s)$ into FTS observations $\{y_d(t_i),i=1,\ldots, n\}_{d=1}^D$ as described in Section \ref{sec::FTS}, we employ a rolling-window segmentation scheme. The $d$-th observed curve with measurement error is $y_d(t_i)=z(s_{(d-1)n(1-\gamma)+i}) + \epsilon_{d,i}$, where $\epsilon_{d,i}\sim N(0,0.05^2)$. The overlap ratio is $\gamma=n^\prime/n \in [0,1)$, where $n^\prime$ is the number of points shared by consecutive curves; $\gamma=0$ corresponds to the standard non-overlapping segmentation.

We consider $\tilde{D}\in \{150,300\}$, $n \in \{100, 300, 500\}$ and $\gamma \in \{0, 0.4, 0.6, 0.8\}$. For each configuration, the data are split into 60\% for training and 40\% for testing. We compare the proposed OLFTSA method, which employs quantile random forests for multivariate conformal prediction, with the FSSA method combined with residual-based bootstrap resampling \citep{Trinka2023FSSA}. Figure \ref{fig:sim_results_D300} summarizes the results for the larger setting $\tilde{D}=300$. Similar patterns hold for $\tilde{D}=150$; the full numerical results for both settings are reported in Tables \ref{Appen:tab:sim_results_D300} and \ref{Appen:tab:sim_results_D150} in the Online Supplementary Materials. FSSA slightly outperforms OLFTSA in point forecasts under the non-overlapping setting with shorter curves ($n\in \{100,300\}$), in agreement with the findings of \citet{Trinka2023FSSA}. As the overlap ratio $\gamma$ increases, however, both RMSE and MAE of OLFTSA decrease, and OLFTSA outperforms FSSA for $\gamma \ge 0.4$. For uncertainty quantification, the OLFTSA conformal prediction sets attain coverage close to the nominal level in all settings, whereas the residual-based bootstrap of FSSA undercovers substantially and produces wider intervals. Increasing the overlap ratio further narrows the OLFTSA intervals without reducing their coverage, and as $\tilde{D}$ increases, the empirical coverage of OLFTSA approaches the nominal level, in line with the asymptotic validity of the procedure. Finally, OLFTSA requires substantially less computation time than FSSA.

\begin{figure}[ht!]
    \centering
    \includegraphics[width=1\linewidth]{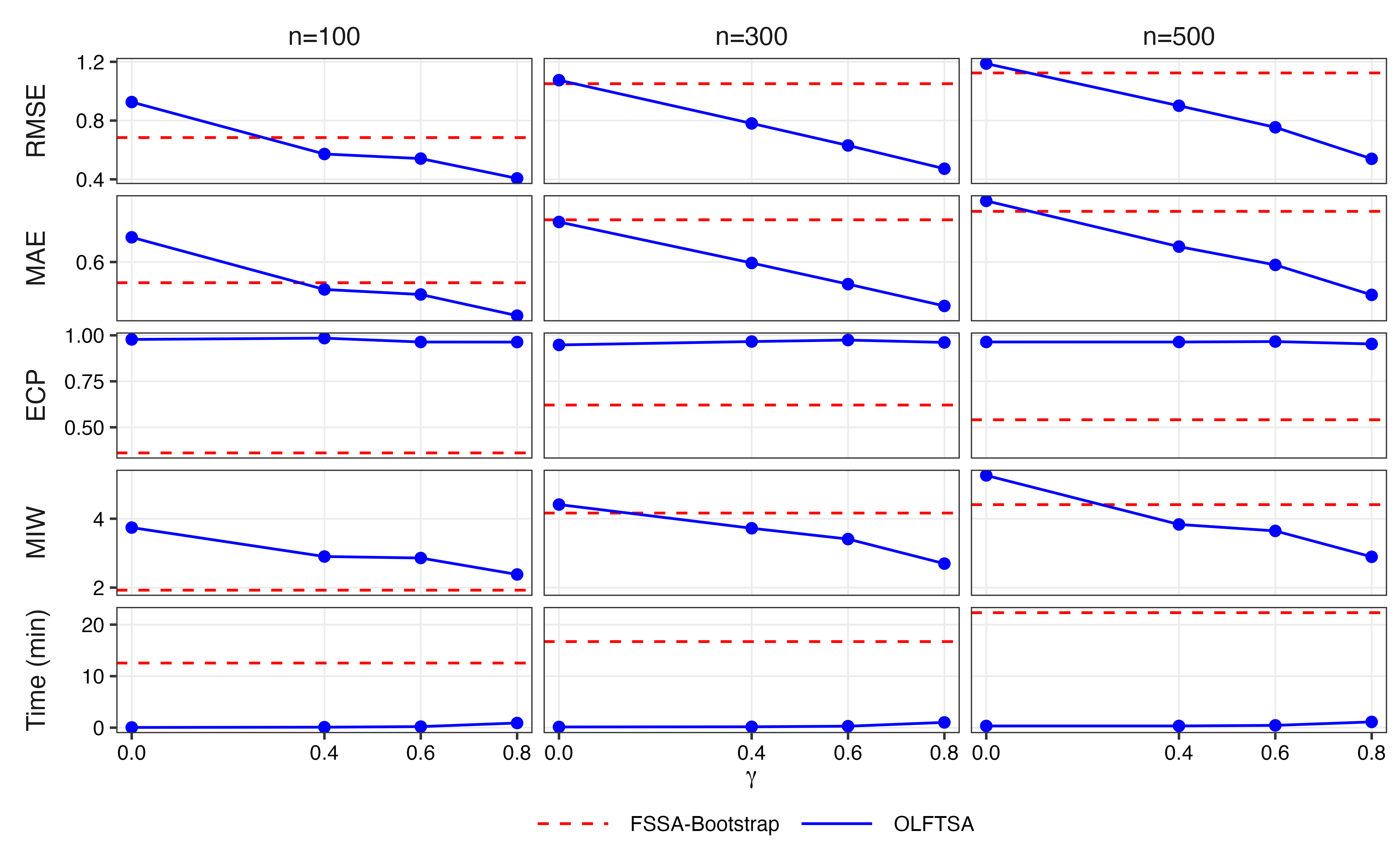}
    \caption{Simulation results for $\tilde{D}=300$ comparing the proposed OLFTSA method under varying overlap ratios $\gamma$ with FSSA-Bootstrap for three curve lengths, $n=100$, $300$, and $500$. Rows display point forecast accuracy (RMSE and MAE), uncertainty quantification performance at nominal coverage level 0.95 (ECP and MIW), and computational time in minutes. Blue solid lines with points represent OLFTSA, while red dashed horizontal lines represent FSSA-Bootstrap. Results are averaged over 50 replications.}
    \label{fig:sim_results_D300}
\end{figure}

\subsubsection{Overlap as a training-data augmentation mechanism}\label{sec::SimAugment}
In the design above, $n$ is held fixed while $\gamma$ varies, so increasing $\gamma$ both shortens the forecast tail $n-n^\prime=n(1-\gamma)$ contributed by each curve and increases the number of patches $D$ available for training and evaluation. We now isolate the second effect. For a fixed raw realization and a fixed evaluation protocol, we ask whether increasing the training-side overlap ratio, which increases the number of training patches at no additional data cost, improves the forecasts.

In each replication, we generate one raw realization of $N=90{,}000$ points ($\tilde{D}=300$, $n=300$) and split it into a training block (60\%) and a test block (40\%). The test block is always segmented with $\gamma_{\text{eval}}=0.8$: each test curve has length $n=300$, and only its non-overlapping tail of $300(1-\gamma_{\text{eval}})=60$ points is scored, so that every test point is evaluated exactly once. Only the training block is re-segmented, with $\gamma_{\text{train}}\in\{0,0.2,0.4,0.6,0.8\}$, which gives $D_{\text{train}}\in \{180,224,299,448,896\}$ training patches from identical data. Because consecutive training patches share up to $240$ of $300$ points at $\gamma_{\text{train}}=0.8$, the VARX stage in this experiment is fitted with a small ridge penalty, selected by 5-fold rolling-origin cross-validation, to keep the estimation stable. Results are averaged over 50 replications.

Figure \ref{fig:sim_augment} shows that the additional training patches improve both point forecasts and uncertainty quantification. From $\gamma_{\text{train}}=0$ to $\gamma_{\text{train}}=0.8$, the test RMSE decreases by $26.8\%$ (from $0.7228$ to $0.5293$) and the MAE by $29.7\%$ (from $0.5775$ to $0.4063$). The conformal intervals also become narrower, with MIW decreasing by $52.5\%$ (from $5.79$ to $2.75$), while the ECP moves from an over-conservative $0.999$ toward the nominal level ($0.968$ at $\gamma_{\text{train}}=0.8$). The narrower intervals therefore reflect improved calibration rather than a loss of coverage. These results are consistent with the role of the training sample size $D$ in Condition \ref{con:4} of Proposition \ref{prop:1}, which bounds the estimation error of the VARX stage at the rate $O_p(\sqrt{K_{\text{tot}}/D})$ in the number of score-space training observations: overlapping segmentation increases $D$ without requiring additional data.

\begin{figure}[ht!]
    \centering
    \includegraphics[width=1\linewidth]{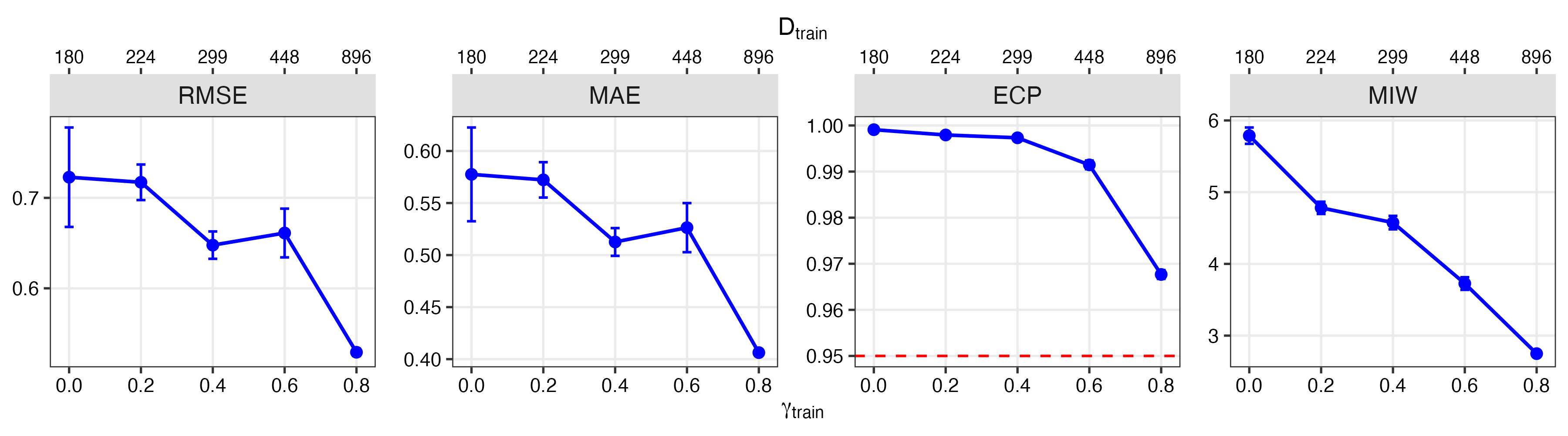}
    \caption{Overlap as training-data augmentation. The raw data, the forecast horizon and the evaluation protocol ($\gamma_{\text{eval}}=0.8$) are fixed; only the training overlap ratio $\gamma_{\text{train}}$ varies, with the corresponding number of training patches $D_{\text{train}}$ shown on the upper axis. Points are means over 50 replications, and error bars are $\pm 1.96$ Monte Carlo standard errors. The red dashed line in the ECP panel marks the nominal level 0.95.}
    \label{fig:sim_augment}
\end{figure}

\section{Application to Geomagnetic Index Forecasting}\label{sec:app}

We forecast the five target indices $p\in\mathcal{P}$ described in Section \ref{sec::Data}: the 3-hourly Kp index, the hourly Dst index and the 5-minute SYM-H, SME and SMR indices. Data from January 1, 1995 to December 31, 2019 (solar cycles 23 and 24) are used for estimation, and data from January 1, 2020 through 2025 (solar cycle 25) form the held-out test set. Each variable $p\in\mathcal{P}\cup\mathcal{Q}$ is segmented into overlapping patches as in Section \ref{ssec:Model}, and we consider two tasks: 6-hour ahead forecasts from 48-hour patches and 24-hour ahead forecasts from 72-hour patches. Consecutive patches are shifted by the forecast horizon. For index $p$, a patch contains $n_p$ grid points; the first $n_p^\prime$ of them are observed at the forecast origin, and the last $n_p-n_p^\prime$ form the forecast block. The overlap ratio, $\gamma=n_p^\prime/n_p$, is therefore $42/48=0.875$ for the 6-hour task and $48/72=0.667$ for the 24-hour task, for every index. Each OMNI variable in $\mathcal{Q}$ has $n_{\mathcal{Q}}=48$ or $72$ hourly grid points in a patch. The number of training patches is $D_{\text{train}}$=36,516 for the 6-hour task and 9,128 for the 24-hour task.

For each variable $p\in\mathcal{P}\cup\mathcal{Q}$, the mean function and the functional PCs are estimated from the training patches, with $K_p$ chosen to explain 99\% of the variance, and are held fixed over the test period. The PC scores follow the VARX(1) model of Section \ref{ssec:Model} with the coefficient matrix $\ve{B}$ restricted as follows. The equation for the scores $\ve{\xi}_{p,d+1}$ of index $p\in\mathcal{P}$ contains an intercept, the $K_p$ current scores $\ve{\xi}_{p,d}$ and a subvector $\ve{\xi}_{(\mathcal{P}\cup\mathcal{Q})\setminus\{p\},d+1}$ of the current scores of the other variables in $(\mathcal{P}\cup\mathcal{Q})\setminus\{p\}$, consisting of the scores whose absolute lag-1 sample cross-correlation with at least one score of index $p$ exceeds 0.3 over the training period; the remaining entries of $\ve{B}$ are set to zero.

We compare OLFTSA with three benchmarks. Two are nonlinear: an LSTM encoder--decoder network \citep{Tan2018Geomagnetic} and the pretrained time series foundation model CHRONOS \citep{ansari2024chronos}. The third is a VARX(1) model fitted to the discretized raw curves, in which the $n_p-n_p^\prime$ values of the forecast block of patch $d+1$ are regressed on the $n_p$ values of patch $d$ of index $p$ and on the $6n_{\mathcal{Q}}$ values of the six OMNI variables over patch $d$. It uses the same patches, lag order, exogenous inputs and training data as OLFTSA and differs from it only in representing each patch by its raw coordinates rather than by PC scores. Accuracy is measured by RMSE, MAE and the Pearson correlation coefficient (Corr) over all test points and over storm-time points. Storm detection is assessed by the Heidke skill score (HSS) and the true skill statistic (TSS) \citep{bloomfield2012toward}. Implementation details, complete results, model dimensions and an extended comparison with the raw VARX(1) are given in Section \ref{sec:S-bench} of the Online Supplementary Materials.

\subsection{Forecast accuracy}\label{sec:app-acc}

Table \ref{tab:acc} summarizes the results over all test points. OLFTSA is more accurate than both nonlinear benchmarks for every index, horizon and metric. At 6 hours, it reduces the RMSE by 17--37\% relative to LSTM and by 10--17\% relative to CHRONOS; for Dst, the RMSE is 8.41 nT, compared with 13.33 nT and 10.10 nT. At 24 hours, the reductions are 6--24\% relative to LSTM and 10--18\% relative to CHRONOS, and the correlation for Dst and SYM-H is 0.75, compared with 0.61--0.66. At these horizons, the predictable component of the indices is dominated by persistence and by an approximately linear response to solar wind driving \citep{burton1975empirical}, which OLFTSA models directly through the lagged scores and the OMNI inputs.

\begin{table}[htbp]
\centering\small\setlength{\tabcolsep}{4pt}
\caption{Forecast accuracy over all test points (2020--2025). Bold: best value in each row for each metric.}\label{tab:acc}
\resizebox{\textwidth}{!}{\begin{tabular}{l ccc ccc ccc ccc}
\toprule
 & \multicolumn{3}{c}{OLFTSA} & \multicolumn{3}{c}{LSTM} & \multicolumn{3}{c}{CHRONOS} & \multicolumn{3}{c}{Raw VARX(1)}\\
\cmidrule(lr){2-4}\cmidrule(lr){5-7}\cmidrule(lr){8-10}\cmidrule(lr){11-13}
Index & RMSE & MAE & Corr & RMSE & MAE & Corr & RMSE & MAE & Corr & RMSE & MAE & Corr\\
\midrule\multicolumn{13}{l}{\textit{6-hour ahead}}\\
Kp & 0.8360 & 0.6330 & 0.7665 & 1.0123 & 0.8442 & 0.7264 & 0.9484 & 0.7099 & 0.6933 & \textbf{0.8173} & \textbf{0.6185} & \textbf{0.7780}\\
Dst & 8.4104 & \textbf{5.0397} & 0.9070 & 13.3263 & 10.0410 & 0.8310 & 10.1035 & 5.8273 & 0.8629 & \textbf{8.1765} & 5.0473 & \textbf{0.9125}\\
SYM-H & 8.8909 & \textbf{5.2990} & 0.9010 & 12.5481 & 9.4569 & 0.8474 & 10.7060 & 6.0783 & 0.8538 & \textbf{8.7743} & 5.3584 & \textbf{0.9037}\\
SMR & 8.9326 & \textbf{5.0798} & 0.8858 & 13.8956 & 10.8873 & 0.8294 & 10.5271 & 5.5673 & 0.8403 & \textbf{8.8484} & 5.0877 & \textbf{0.8881}\\
SME & \textbf{162.9377} & \textbf{106.1887} & \textbf{0.6525} & 222.0116 & 190.3414 & 0.5874 & 180.7999 & 107.2291 & 0.5774 & 163.8911 & 106.6430 & 0.6502\\
\midrule\multicolumn{13}{l}{\textit{24-hour ahead}}\\
Kp & \textbf{1.0520} & \textbf{0.8027} & \textbf{0.5899} & 1.1505 & 0.9082 & 0.4630 & 1.1682 & 0.8729 & 0.4909 & 1.0747 & 0.8140 & 0.5741\\
Dst & \textbf{13.0899} & \textbf{8.0703} & \textbf{0.7555} & 15.9709 & 10.0122 & 0.6247 & 15.0583 & 9.0356 & 0.6629 & 13.4308 & 8.3899 & 0.7433\\
SYM-H & \textbf{13.4627} & \textbf{8.2747} & \textbf{0.7538} & 17.7221 & 11.2564 & 0.6147 & 16.4369 & 9.6024 & 0.6120 & 14.7424 & 9.3288 & 0.7028\\
SMR & \textbf{13.5061} & \textbf{8.0424} & \textbf{0.7126} & 16.4162 & 9.8410 & 0.5713 & 16.2051 & 8.8930 & 0.5631 & 14.5130 & 8.9296 & 0.6697\\
SME & \textbf{196.2860} & \textbf{133.4225} & \textbf{0.4235} & 209.2637 & 152.1024 & 0.2849 & 217.5409 & 134.2652 & 0.3201 & 209.3869 & 143.3890 & 0.3655\\
\bottomrule
\end{tabular}}
\end{table}

The raw VARX(1) shares the linear dynamics and the inputs of OLFTSA, so the comparison between the two isolates the effect of the functional representation. At 6 hours the two models are nearly equivalent: the raw VARX(1) has a 0.9--2.8\% lower RMSE for Kp, Dst, SYM-H and SMR, OLFTSA has the lower RMSE for SME and the lower MAE for Dst, SYM-H, SMR and SME, and the correlations differ by at most 0.012. At 24 hours, OLFTSA is more accurate on all 15 comparisons, with RMSE reductions of 2.1\% (Kp), 2.5\% (Dst), 8.7\% (SYM-H), 6.9\% (SMR) and 6.3\% (SME) and an increase in correlation from 0.70 to 0.75 for SYM-H.

The dependence of the gain on the horizon and on the index is explained by the dimensions of the two linear models (Table \ref{tab:dim}). The raw VARX(1) has $n_p-n_p^\prime$ equations, one for each grid point of the forecast block, and $1+n_p+6n_{\mathcal{Q}}$ predictors per equation, so that it has $
\nu^{\text{raw}}_p=(n_p-n_p^\prime)(1+n_p+6n_{\mathcal{Q}})$
regression coefficients; its dimension therefore grows with the sampling rate and the patch length. OLFTSA has $K_p$ equations and $1+K_p+K_{(\mathcal{P}\cup\mathcal{Q})\setminus\{p\}}$ predictors per equation, so that it has $\nu^{\text{OL}}_p=K_p(1+K_p+K_{(\mathcal{P}\cup\mathcal{Q})\setminus\{p\}})$ regression coefficients. Its dimension is determined by the number of PCs needed to represent the curves rather than by the grid: across all indices and both horizons, the equations of OLFTSA have 321--483 predictors, whereas those of the raw VARX(1) have 305--1{,}297.

\begin{table}[htbp]
\centering\small\setlength{\tabcolsep}{5pt}
\caption{Dimensions of the forecasting models. $n_p$: grid points per patch; $K_p$: retained PCs; Raw and OLFTSA: predictors per equation, $1+n_p+6n_{\mathcal{Q}}$ and $1+K_p+K_{(\mathcal{P}\cup\mathcal{Q})\setminus\{p\}}$; Ratio: $\nu^{\text{raw}}_p/\nu^{\text{OL}}_p$, the number of regression coefficients of the raw VARX(1) divided by that of OLFTSA. The number of training patches is $D_{\text{train}}=36{,}516$ at 6 hours and $9{,}128$ at 24 hours. Full counts are given in Table \ref{tab:S-dim} of the Online Supplementary Materials.}\label{tab:dim}
\begin{tabular}{l rrrcr rrrcr}
\toprule
 & \multicolumn{5}{c}{6-hour ahead} & \multicolumn{5}{c}{24-hour ahead}\\
\cmidrule(lr){2-6}\cmidrule(lr){7-11}
 & & & \multicolumn{2}{c}{Predictors} & & & & \multicolumn{2}{c}{Predictors} & \\
\cmidrule(lr){4-5}\cmidrule(lr){9-10}
Index & $n_p$ & $K_p$ & Raw & OLFTSA & Ratio & $n_p$ & $K_p$ & Raw & OLFTSA & Ratio\\
\midrule
Kp & 16 & 6 & 305 & 321 & 0.3167 & 24 & 10 & 457 & 420 & 0.8705\\
Dst & 48 & 16 & 337 & 348 & 0.3631 & 72 & 24 & 505 & 483 & 1.0455\\
SYM-H & 576 & 36 & 865 & 341 & 5.0733 & 864 & 52 & 1{,}297 & 474 & 15.1548\\
SMR & 576 & 43 & 865 & 347 & 4.1740 & 864 & 63 & 1{,}297 & 478 & 12.4041\\
SME & 576 & 139 & 865 & 349 & 1.2838 & 864 & 200 & 1{,}297 & 439 & 4.2544\\
\bottomrule
\end{tabular}
\end{table}

At 6 hours, dimension reduction offers little benefit. With $D_{\text{train}}=36{,}516$, the ratio of predictors per equation to training patches is at most 0.024 for the raw VARX(1), so the unrestricted regression is well determined. In addition, 42 of the 48 hours of the target patch ($n_p^\prime/n_p=0.875$) are observed, and the forecast depends mainly on the most recent observations, which the raw coordinates retain exactly and a truncated PC expansion approximates less accurately near the end of the patch. This small truncation bias accounts for the slight advantage of the raw VARX(1), which is largest for Kp and Dst. For these coarse indices, a patch contains only $n_p=16$ and $48$ grid points, so the functional representation provides little compression, and the cross-variable scores give OLFTSA about three times as many coefficients as the raw VARX(1).

At 24 hours, the balance reverses. For the 5-minute indices, each raw equation has 1{,}297 predictors estimated from $D_{\text{train}}=9{,}128$ training patches, and the raw VARX(1) has $\nu^{\text{raw}}_p=373{,}536$ coefficients, whereas OLFTSA represents an 864-point patch of SYM-H by $K_p=52$ scores and its equations have 474 predictors, including the scores of the other indices. Figure \ref{fig:dim}(a) shows that the gain of OLFTSA follows this reduction in dimension. At 24 hours, the ordering of the RMSE gains, SYM-H $>$ SMR $>$ SME $>$ Dst $>$ Kp, coincides with the ordering of the coefficient ratios $\nu^{\text{raw}}_p/\nu^{\text{OL}}_p$ (15.15, 12.40, 4.25, 1.05 and 0.87), and the gain is smallest for Kp and Dst, for which the two models have similar dimensions. Consistently, the accuracy of OLFTSA deteriorates more slowly with the horizon (Figure \ref{fig:dim}b): from 6 to 24 hours its RMSE increases by 20--56\%, compared with 28--68\% for the raw VARX(1).

The same property allows OLFTSA to incorporate information from other indices. Adding the raw curves of the other four indices to the raw VARX(1) would give about 4{,}750 predictors per equation for the 5-minute indices at 24 hours, about half of $D_{\text{train}}$. OLFTSA includes 101--324 lagged scores of the other indices and 81--134 OMNI scores in $\ve{\xi}_{(\mathcal{P}\cup\mathcal{Q})\setminus\{p\},d}$ while keeping the ratio of predictors per equation to training patches near 0.05. The functional representation also yields smooth forecast trajectories and the conformal prediction regions of Section \ref{sec:UQ}.

\begin{figure}[htbp]
\centering
\includegraphics[width=\textwidth]{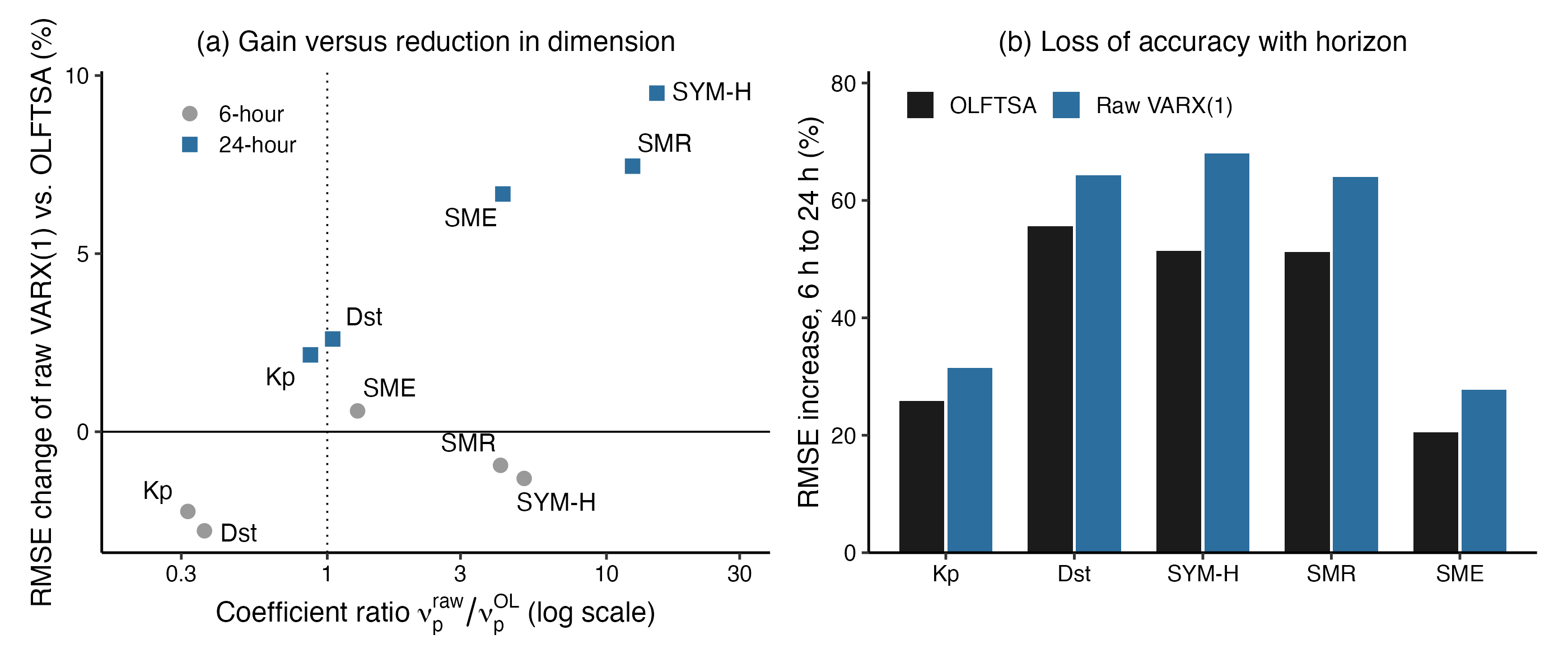}
\caption{Effect of the functional representation. (a) Percentage change in RMSE of the raw VARX(1) relative to OLFTSA against the ratio $\nu^{\text{raw}}_p/\nu^{\text{OL}}_p$ of their numbers of regression coefficients (log scale); positive values indicate that OLFTSA is more accurate. (b) Percentage increase in RMSE from the 6-hour to the 24-hour horizon.}\label{fig:dim}
\end{figure}

\subsection{Storm-time performance}\label{sec:app-storm}

Table \ref{tab:storm} reports the RMSE over storm-time points and the skill in detecting them, where a time point is classified as an event when the observed or predicted value crosses the threshold in parentheses. At 24 hours, OLFTSA attains the highest HSS and TSS of all methods for Dst, SYM-H and SMR. Relative to LSTM and CHRONOS, it reduces the storm-time RMSE of these three indices by 19--28\% and increases the storm-time correlation from 0.27--0.45 to 0.63--0.66. Relative to the raw VARX(1), its storm-time RMSE is lower for Kp, SYM-H and SME and within 2\% for Dst and SMR, and its false alarm ratio is lower for all three ring-current indices.

\begin{table}[htbp]
\centering\small\setlength{\tabcolsep}{3.4pt}
\caption{Storm-time RMSE and event skill. Events are defined by the thresholds in parentheses (nT for Dst, SYM-H, SMR and SME). Bold: best value in each row for each metric (lowest RMSE, highest HSS and TSS).}\label{tab:storm}
\resizebox{\textwidth}{!}{\begin{tabular}{l ccc ccc ccc ccc}
\toprule
 & \multicolumn{3}{c}{OLFTSA} & \multicolumn{3}{c}{LSTM} & \multicolumn{3}{c}{CHRONOS} & \multicolumn{3}{c}{Raw VARX(1)}\\
\cmidrule(lr){2-4}\cmidrule(lr){5-7}\cmidrule(lr){8-10}\cmidrule(lr){11-13}
Index & RMSE & HSS & TSS & RMSE & HSS & TSS & RMSE & HSS & TSS & RMSE & HSS & TSS\\
\midrule\multicolumn{13}{l}{\textit{6-hour ahead}}\\
Kp ($\geq 6$) & 2.3114 & \textbf{0.4591} & 0.3750 & 2.2560 & 0.3559 & 0.2377 & 2.6060 & 0.3101 & 0.2025 & \textbf{2.0396} & 0.4581 & \textbf{0.3899}\\
Dst ($\leq -100$) & 61.5742 & \textbf{0.6919} & 0.6108 & 57.3727 & 0.6833 & \textbf{0.7213} & 78.8084 & 0.6654 & 0.5371 & \textbf{54.8717} & 0.6813 & 0.6435\\
SYM-H ($\leq -100$) & 60.0679 & 0.6938 & 0.6101 & 59.1844 & 0.6855 & 0.6385 & 76.3472 & 0.6159 & 0.4771 & \textbf{53.2533} & \textbf{0.7245} & \textbf{0.6666}\\
SMR ($\leq -100$) & 68.3490 & 0.6685 & 0.5751 & 63.9960 & \textbf{0.6812} & \textbf{0.6663} & 83.3739 & 0.5952 & 0.4675 & \textbf{63.9556} & 0.6744 & 0.6089\\
SME ($\geq 100$) & \textbf{190.8735} & 0.3587 & 0.3199 & 216.5027 & 0.0000 & 0.0000 & 219.7715 & \textbf{0.4460} & \textbf{0.4404} & 192.0404 & 0.3579 & 0.3202\\
\midrule\multicolumn{13}{l}{\textit{24-hour ahead}}\\
Kp ($\geq 6$) & \textbf{3.5704} & 0.1815 & 0.1143 & 4.1679 & 0.0000 & 0.0000 & 3.8901 & 0.1051 & 0.0568 & 3.6705 & \textbf{0.2455} & \textbf{0.1677}\\
Dst ($\leq -100$) & 91.9838 & \textbf{0.4890} & \textbf{0.3834} & 125.3727 & 0.2659 & 0.1609 & 112.9676 & 0.4062 & 0.2951 & \textbf{90.2725} & 0.4500 & 0.3584\\
SYM-H ($\leq -100$) & \textbf{91.1846} & \textbf{0.4563} & \textbf{0.3571} & 126.8560 & 0.2434 & 0.1426 & 123.5225 & 0.3541 & 0.2332 & 93.7005 & 0.4261 & 0.3353\\
SMR ($\leq -100$) & 100.2500 & \textbf{0.4422} & \textbf{0.3285} & 135.8564 & 0.1146 & 0.0646 & 133.0117 & 0.2590 & 0.1712 & \textbf{100.2296} & 0.4179 & 0.3267\\
SME ($\geq 100$) & \textbf{225.4232} & 0.1829 & 0.1562 & 230.5688 & 0.0603 & 0.0494 & 259.0995 & \textbf{0.2771} & \textbf{0.2672} & 238.4273 & 0.1656 & 0.1460\\
\bottomrule
\end{tabular}}
\end{table}

At 6 hours, the raw VARX(1) has a 6--12\% lower storm-time RMSE than OLFTSA for Kp, Dst, SYM-H and SMR. Intense main-phase decreases are rare and sharp, contribute little to the variance explained by the leading PCs and are therefore partly attenuated by the truncated expansion, whereas the raw regression propagates the most recent depression directly. This attenuation makes OLFTSA more conservative rather than less skillful: it attains the highest HSS of all methods for Kp and Dst, and a lower false alarm ratio than the raw VARX(1) for Kp, Dst, SYM-H and SMR. The nonlinear benchmarks are less reliable at the extremes. The LSTM classifies every 6-hour SME forecast as active and never forecasts Kp $\geq 6$ at 24 hours, and CHRONOS, which does not observe the solar wind, detects 47--54\% of 6-hour ring-current storm points, compared with 58--61\% for OLFTSA. Two results favor other methods: at 24 hours, the raw VARX(1) has the higher HSS and TSS for Kp $\geq 6$, which occurs at 0.8\% of test points, and CHRONOS has the highest HSS and TSS for SME at both horizons.

\subsection{Illustration and uncertainty quantification}\label{sec:app-case}

Figure \ref{fig:case} shows 24-hour ahead forecasts of Dst over April--May 2024, a period that includes the intense geomagnetic storm of 10--11 May, and of SME over January--February 2000, a period with several intervals of substorm activity. The accuracy measures in the figure are computed over the displayed periods. Over these periods, OLFTSA attains the lowest RMSE for both indices, although for Dst its RMSE is close to that of the raw VARX(1) (29.4074 and 29.4712 nT), and it attains the highest correlation for SME (0.3691, compared with 0.1362--0.2630). The forecast trajectories show differences between the methods that these summary measures do not capture.

For Dst, the four methods differ mainly during the storm. OLFTSA reproduces the timing of the sudden main-phase decrease and then follows the gradual recovery with a smooth trajectory that stays close to the observations until quiet conditions return, although it underestimates the depth of the minimum, which is approximately $-400$ nT. The raw VARX(1) also captures the onset, but its recovery is interrupted by a secondary decrease that is not present in the observations. The CHRONOS forecast is piecewise flat: it remains near $-170$ nT during the main phase and changes level in discrete steps during quiet periods. The LSTM forecast misses most of the main-phase decrease and reaches its lowest values only later, during the recovery phase. During the quiet intervals before and after the storm, OLFTSA also tracks the smaller disturbances of Dst.

For SME, no method reproduces the individual substorm peaks 24 hours in advance, and the correlations are correspondingly modest. The OLFTSA forecast nevertheless separates active from quiet conditions: it rises during the intervals of repeated substorms, in late January, around 8 February and around 18--20 February, and returns to the quiet-time level between them. The LSTM forecast oscillates with a nearly constant period and amplitude regardless of the observed activity, and CHRONOS remains close to the quiet-time level except for a peak in late February that does not correspond to an active interval. The raw VARX(1) forecast responds to activity but fluctuates strongly and falls well below zero around 8 February, although SME is non-negative by definition. These patterns are consistent with the role of the functional representation discussed in Section \ref{sec:app-acc}: projecting each input curve onto its leading $K_p$ components removes high-frequency variation, and the smaller number of coefficients is estimated more stably than in the raw VARX(1), which yields forecasts that follow the level of activity without spurious oscillations.

OLFTSA also provides prediction regions through the conformal procedure of Section \ref{sec:UQ}, calibrated at the nominal level 0.95. For 6-hour ahead Dst forecasts (Figure \ref{fig:case}, bottom), the intervals attain an empirical coverage probability of 0.924 with a mean interval width of 28.1 nT. At 24 hours, the empirical coverage is 0.892 with a mean width of 56.6 nT for Dst and 0.897 with a mean width of 58.8 nT for SYM-H.

\begin{figure}[htbp]
\centering
\IfFileExists{figures/DstSME_24h.png}{\includegraphics[width=\textwidth]{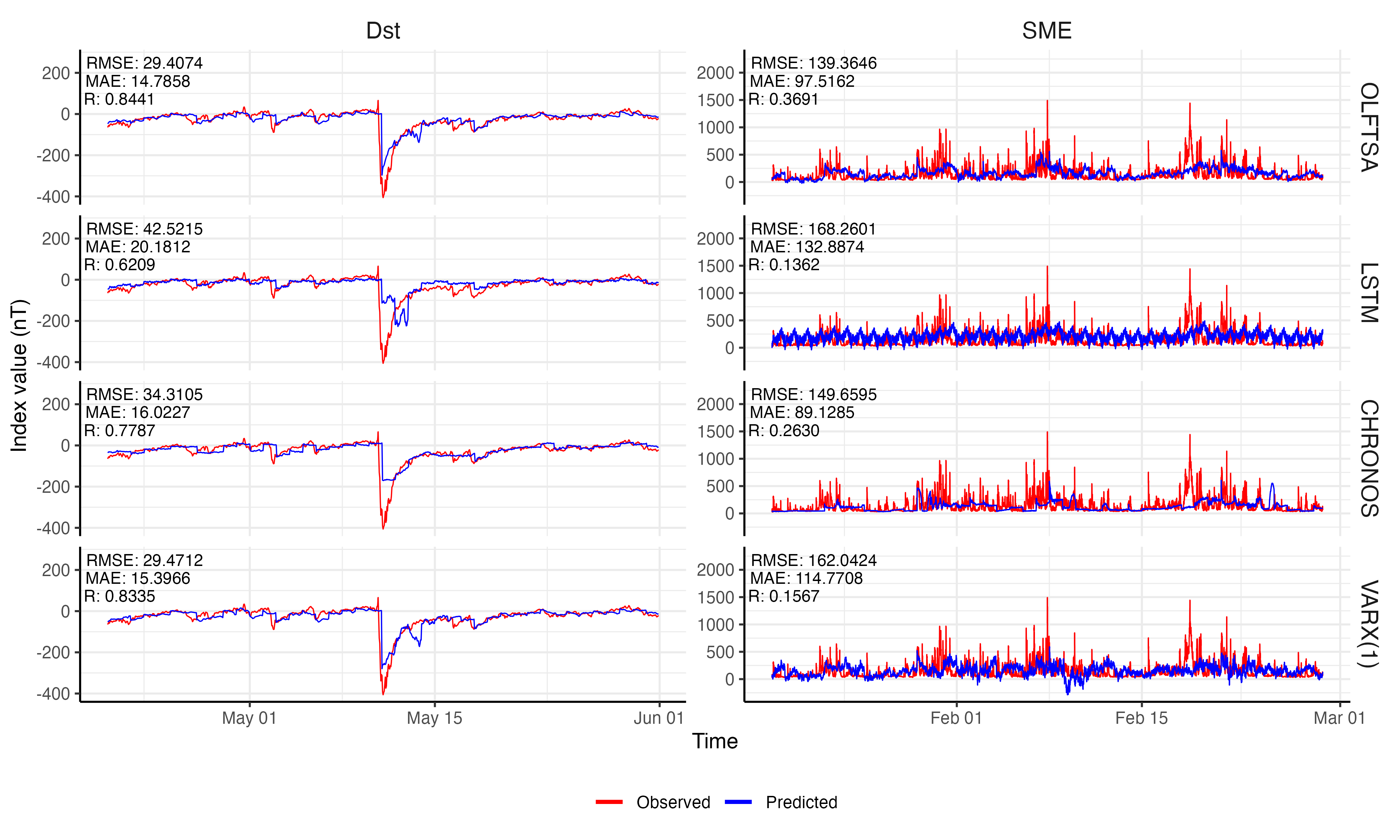}}{\fbox{\parbox[c][1in][c]{0.9\textwidth}{\centering DstSME\_24h.png}}}\\[4pt]
\IfFileExists{figures/UQ-Dst6hr.png}{\includegraphics[width=\textwidth]{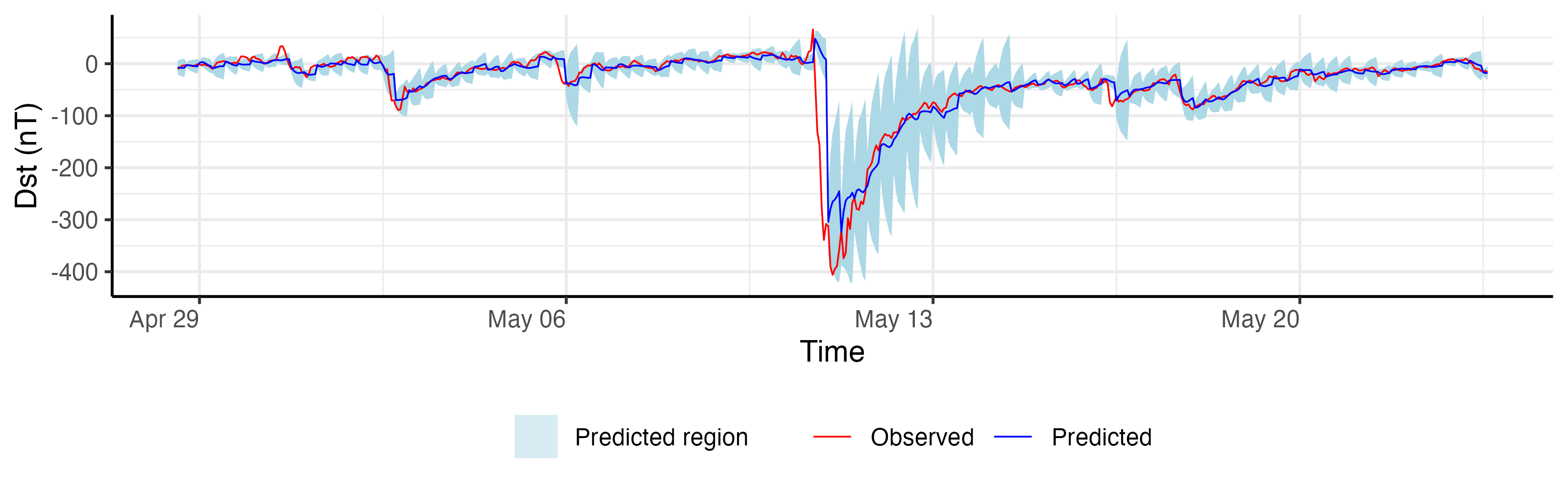}}{\fbox{\parbox[c][0.8in][c]{0.9\textwidth}{\centering UQ-Dst6hr.png}}}
\caption{Top: 24-hour ahead forecasts of Dst (left; April--May 2024, including the storm of 10--11 May) and SME (right; January--February 2000). Observed values are shown in red and forecasts in blue for OLFTSA, LSTM, CHRONOS and the raw VARX(1) (rows, top to bottom); RMSE, MAE and Corr are computed over the displayed periods. Bottom: 6-hour ahead OLFTSA forecasts of Dst with conformal prediction intervals at the nominal level 0.95.}\label{fig:case}
\end{figure}

\section{Conclusion}\label{sec:Disc}
This paper introduces an overlapping functional time series analysis (OLFTSA) framework integrated with conformal prediction for multivariate forecasting of geomagnetic indices. The methodology combines extended temporal windowing, FPCA, and multivariate time series modeling to accommodate heterogeneous indices in space weather prediction.

Representing 48--72 hour patches in functional form preserves temporal continuity, while FPCA-based dimensionality reduction enables joint forecasting of the Kp, Dst, SYM-H, SME, and SMR indices driven by solar wind parameters. In the applications considered here, OLFTSA attains lower forecast errors and higher correlations than the LSTM and CHRONOS benchmarks, particularly at the 6-hour horizon, with comparable performance at 24 hours.

The overlapping rolling-window construction serves as a training-data augmentation mechanism. Holding the raw data, forecast horizon, and evaluation protocol fixed, increasing the training-side overlap ratio from $\gamma_{\mathrm{train}}=0$ to $0.8$ raises the number of training patches from $180$ to $896$ at no additional data cost, reducing RMSE by $26.8\%$, MAE by $29.6\%$, and mean conformal interval width by $52.5\%$ in our simulations. This is consistent with the role of the training sample size $D$ in the VARX-stage estimation rate of Proposition \ref{prop:1}. The trade-off is that longer windows require retaining more functional PCs, which increases the computational cost of the multivariate forecasting step.

The adapted conformal prediction procedure provides asymptotically valid, distribution-free coverage (Theorem \ref{Thm:coverage}), and its empirical coverage probabilities are close to the nominal $95\%$ level in both simulations and applications. This gives a reliable mechanism for uncertainty quantification in operational settings. The framework is also computationally efficient: each forecast step takes seconds, compared with the several minutes typically required by the deep-learning baselines considered here.

Several considerations arise for operational implementation. SuperMAG SME and SMR indices, despite their utility for auroral monitoring, exhibit data latencies exceeding two months, which limits their direct use in real-time forecasting. OMNI solar wind parameters are available in near real time but incur approximately 1-hour upstream propagation delays. Conformal prediction intervals also display discontinuities at patch boundaries, which may complicate interpretation for end users. Future work could address these issues in several directions. Real-time systems might replace the delayed high-latitude indices with WDC-Kyoto AE data and incorporate physics-informed solar wind nowcasting (e.g., ORSMAR models) to extend forecast horizons. Methodological developments could pursue continuous-time conformal procedures to obtain smoother prediction bands, and could accommodate bounded supports (e.g., Kp in $[0,9]$) or heavy-tailed behavior through appropriate score transformations. Finally, extensions to full magnetometer networks and upstream CME propagation models could be used to study the feasibility of spatially resolved, multi-day forecasts within a similar functional and conformal framework.

\section*{Data and Code Availability}
The data and code are available at \url{https://doi.org/10.5281/zenodo.20274712}. 

\if1\anon
{
\section*{Acknowledgments}
Chen's work is supported by funds from the the NASA Space Weather Center of Excellence program under award No. 80NSSC23M0191 and No. 80NSSC23M0192.
} \fi

\bibliography{Bibliography-MM-MC}
\end{document}